\documentclass[lettersize,journal]{IEEEtran}
\usepackage{amsmath,amsfonts}
\usepackage{algorithmic}
\usepackage{algorithm}
\usepackage{array}
\usepackage[caption=false,font=normalsize,labelfont=sf,textfont=sf]{subfig}
\usepackage{textcomp}
\usepackage{stfloats}
\usepackage{url}
\usepackage{verbatim}
\usepackage{graphicx}
\usepackage{cite}
\usepackage{amssymb}
\usepackage{bm} 
\usepackage[thmmarks,amsmath]{ntheorem}
\usepackage{cuted}\stripsep22pt
\begin{document}
\newtheorem*{proof}{Proof}
\newtheorem{lemma}{Lemma}
\newtheorem{proposition}{Proposition}
\newtheorem{corollary}{Corollary}
\newtheorem{definition}{Definition}
\newtheorem{theorem}{Theorem}
\newtheorem{remark}{Remark}
\newtheorem{example}{Example}
\theoremsymbol{$\blacksquare$}
\def\QEDclosed{\mbox{\rule[0pt]{1.3ex}{1.3ex}}} 
\def\QEDopen{{\setlength{\fboxsep}{0pt}\setlength{\fboxrule}{0.2pt}\fbox{\rule[0pt]{0pt}{1.3ex}\rule[0pt]{1.3ex}{0pt}}}} 
\def\QED{\QEDclosed} 
\def\endproof{\hspace*{\fill}~\QED\par\endtrivlist\unskip}

\title{Target Controllability of Multiagent Systems under Directed Weighted Topology}

\author{Yanan Ji, Zhijian Ji, Yungang Liu, and Chong Lin,

\thanks{This work was supported by the National Natural Science Foundation of China (Grant Nos. 62033007 and 61873136),
Taishan Scholars Climbing Program of Shandong Province of China and Taishan Scholars Project of Shandong Province of China (No. ts20190930). }
\thanks{Yanan Ji, Zhijian Ji (corresponding author), and Chong Lin are with Institute of Complexity Science, College of Automation, Qingdao University, and Shandong Key Laboratory of Industrial Control Technology, Qingdao, Shandong, 266071, China (corresponding author to provide e-mail: jizhijian@pku.org.cn). }
\thanks{Yungang Liu is with School of Control Science and
Engineering, Shandong University, Jinan, Shandong, 250061, China (e-mail: lygfr@sdu.edu.cn).}
\thanks{Corresponding author: Zhijian Ji (e-mail: jizhijian@pku.org.cn).}}

\maketitle

\begin{abstract}
In this paper, the target controllability of multiagent systems under directed weighted topology is studied. A graph partition is constructed, in which part of the nodes are divided into different cells, which are selected as leaders. The remaining nodes are divided by maximum equitable partition. By taking the advantage of reachable nodes and the graph partition, we provide a necessary and sufficient condition for the target controllability of a first-order multiagent system. It is shown that the system is target controllable if and only if each cell contains no more than one target node and there are no unreachable target nodes, with $\delta-$reachable nodes belonging to the same cell in the above graph partition. By means of controllability decomposition, a necessary and sufficient condition for the target controllability of the system  is given, as well as a target node selection method  to ensure the target controllability. In a high-order multiagent system, once the topology, leaders, and target nodes are fixed, the target controllability of the high-order multiagent system is shown to be the same to the first-order one.  This paper also considers a general linear system. If there is an independent strongly connected component that contains only target nodes and the general linear system is target controllable, then graph $\mathcal{G}$ is leader-target follower connected.
\end{abstract}


\section{Introduction}
\label{sec1}
\IEEEPARstart{I}{n} the 21st century, with the rapid development of artificial intelligence, the related researches of multiagent systems have attracted extensive attention from all walks of life. Compared with single agent-based models, multiagent systems can complete more complex tasks. For multiagent systems, each agent possesses many attributes such as autonomy, interaction, coordination, and intelligence, which, along with their mutual cooperation and collaboration, produces the ability of “$1+1>2$” in the whole system. In view of the unique advantages of distributed multiagent systems in flexibility, innovation, adaptability, and self-organization, people have conducted more in-depth research on it \cite{ref1,ref2,ref3,ref4,ref29,ref31,ref32,ref33,ref34,ref35,ref36,ref37,ref38,ref39}.  At present, multiagent systems have been widely applied in military, engineering, transportation, and other fields.

The research on controllability of multiagent systems is of great significance in practical production and life, which provides important theoretical guidance in the fields of group formation control, transmission network control, social public opinion control, and infectious disease prevention. As early as the 1960s, R.E.Kalman put forward the concept of controllability and gave the classical algebraic Kalman criterion for controllability in modern control theory \cite{ref12}. The controllability of multiagent systems was first proposed by Tanner in 2004, and a basic framework was presented to consider controllability \cite{ref13}. The controllability of multiagent systems refers to that all agents can be set to specified state from any initial state in limited time. 

In practical application, many systems are composed of  a large number of individuals and complex structures, which makes the control of the entire system difficult. However,  target controllability only requires that a few nodes are controllable. Many scholars have shown great interest to target controllability. For example, A. Trontis et al. studied the target control for linear hybrid systems \cite {ref14}. Ali Ebrahimi et al. made significant progress in target controllability with minimal mediators in complex biological networks \cite{ref15}. Henk J. van Waarde et al. proposed a distance-based approach to strong target control of dynamical networks \cite{ref16}, etc. The research into target controllability is ongoing and has been developed in biology \cite{ref15}, control science and engineering \cite {ref17}\cite{ref18}, medicine \cite{ref19}, and other fields. At present, the research on target controllability mostly focuses on complex networks \cite{ref20,ref21,ref22}, but there are a few results about multiagent systems. In this paper, we study the target controllability of multi-agent systems under the directed weighted topology, and improve the proof of our conclusions in  \cite{ref40}. A necessary condition for the target controllability of a first-order multiagent system is derived, that is,  if the multiagent system with directed weighted topology is target controllable, then each node in the target set is reachable from the leader set. We give a necessary and sufficient condition for the target reachable node, that is, each target node is reachable from the leader set if and only if there are no zero rows in the target controllability matrix. Based on controllability decomposition, a necessary and sufficient condition for the target controllability is proposed, as well as a target node selection method  to ensure the target controllability. In a high-order multiagent system, once the topology, leaders, and target nodes are fixed, the target controllability of the high-order multiagent system is the same to the first-order one. We also study a general linear system.  Based on independent strongly connected components, a necessary condition for the target controllability of the system is given, that is, if there is an independent strongly connected component that contains only target nodes and the system  is target controllable, the graph $\mathcal{G}$ is leader-target follower connected.

The rest of this paper is organized as follows: Section \ref{sec2} shows preliminaries and some notations. Section \ref{sec3} is about the target controllability of a first-order multiagent system. In Section \ref{sec4}, we study the target controllability of a high-order system.  Section \ref{sec5} presents the target controllability of a general linear system. Illustrative examples are presented in Section \ref{sec7}. Finally, the conclusion is arranged in Section \ref{sec8}.

\section{Preparatory Knowledge}
\label{sec2}
$R$, $Z$,  $Z^{+}$, and $\mathbb{N}$ represent real number set, integer set, positive integer set, and natural number set, respectively. $S$ represents a matrix of the special form (Definition \ref{def8}). $I_n$ stands for a $n\times n$ dimensional identity matrix. $0_n$ ($0_{n\times m}$) represents the block of $n\times n$($n\times m$) dimension with all elements of 0,  which can be abbreviated as 0. $\mathbb{R}^n$ and $\mathbb{R}^{n\times m}$ are Euclidean space of $n$ dimensional and real matrix set, respectively. Let $X$ and $Y$ be two subsets, and $X\backslash Y$ represent the set $\left\{ {x|x \in X,x \notin Y} \right\}$.

Consider a multiagent system represented by a directed weighted graph $\mathcal{G}=\{\mathcal{V},\mathcal{E},\mathbf{A}\}$, where $\mathcal {V} = \{v_1, v_2, \cdots, v_n \} $, $\mathcal{E}\subseteq\mathcal{V}\times\mathcal{V}$ are the node set and edge set, respectively. $e_{ij}$ in a directed graph indicates  the edge pointed from $v_j$ to $v_i$,  where $v_j$ is called the parent node of the child node $v_i$, and $v_j$ is the neighbor of $v_i$.  $\mathcal{N}_i=\{v_j|v_j\in\mathcal{V},e_{ij}\in \mathcal{E},i\neq j\}$ denotes the neighbor set of $v_i$. $(v_i,v_i)$ represents a self-loop on $v_i$. A graph is called simple if there are no multiple edges and self-loops.  Self-loop and multiple edge structures are not considered in this paper. Reassigning the weight to $1$, a path of length $m$ in $\mathcal{G}$ is given by a sequence of distinct vertices $v_{i_0},v_{i_1},\cdots,v_{i_m}$.
The adjacency matrix $\mathbf{A}=[a_{ij}]\in \mathbb{R}^{n\times n}$ of a directed weighted graph $\mathcal{G}$ describes the information exchange between nodes of $\mathcal{G}$, which is defined as
$$\begin{aligned}
a_{ij}=
\begin{cases}
w_{ij}, &(j,i)\in \mathcal{E},i\neq j\cr
0, &i= j ,\cr \end{cases}
\end{aligned}
$$
where $w_ {ij}$ represents the weight of the edge.
Laplacian matrix $L=[l_{ij}]\in \mathbb{R}^{n\times n}$ of  $\mathcal{G}$ is defined as
$$\begin{aligned}
l_{ij}=
\begin{cases}
\sum\limits_{k=1,k\neq i}^{n}{ a_{ik}}, &i=j \cr
-a_{ij}, &i\neq j. \cr \end{cases}
\end{aligned}
$$
The degree matrix is a diagonal matrix $D(\mathcal{G})=diag\{d_1,d_2,\cdots,d_n\}$, where $d_{i}=\sum_{j \in {\mathcal{N}_i}}^{}{ a_{ij}}$. Obviously, the Laplacian matrix can also be expressed as
$$\begin{aligned}
L(\mathcal{G}) = D(\mathcal{G})-\mathbf{A}(\mathcal{G}).
\end{aligned}$$

\begin{definition}\cite{ref25}
The partition $\pi=\{\mathcal{C}_1,\mathcal{C}_2,\ldots,\mathcal{C}_\jmath\}$ of a node set $\mathcal{V}$ is called an equitable partition (EP) if for any $v_1,v_2\in \mathcal{C}_i$, $v_3,v_4\in \mathcal{C}_j$, and $i,j=1,2,\ldots,\jmath$, the following equation holds

$$
\sum_{v_3 \in \mathcal{C}_{j}, v_3 \in \mathcal{N}_{v_1}} a_{v_1 v_3}=\sum_{v_4 \in \mathcal{C}_{j}, v_4 \in \mathcal{N}_{v_2}} a_{v_2 v_4}.
$$
\end{definition}
\begin{definition}\cite{ref23}
Consider a system,
\begin{align}\label{16}
\begin{cases}
\dot{x}(t)=Ax(t)+Bu(t),\\
y(t)=\mathbf{C}x(t),
\end{cases}
\end{align}
where ${x}(t)\in{\mathbb{R}}^{n}$ is the plant state, $u(t)\in{\mathbb{R}}^{l}$ is the control input,  and $y(t)\in{\mathbb{R}}^{p}$ is the output.
System \eqref{16} is said to be output controllable on $[t_0, t_f]$, if there is a control input $u(t)$, so that for given $t_0$ and $t_f$, each initial state $x(t_0)$ can be transferred to any terminal state $y(t_f)$.
\end{definition}

\begin{lemma}\label{lem4}
\cite{ref23} System \eqref{16} is output controllable if and only if $\hat{W}=[\mathbf{C}B,\mathbf{C}AB,\mathbf{C}A^2B,\cdots,\mathbf{C}A^{n-1}B]$ is  full row rank, where  $\hat{W}$ is called the output controllability matrix. 
\end{lemma}

\begin{definition}
\cite{ref30} A graph ${\mathcal{G}}^{\prime}=\left\{{\mathcal{V}}^{\prime},{\mathcal{E}}^{\prime},A^{\prime}\right\}$ is called an induced subgraph of the graph $\mathcal{G}=\{\mathcal{V},\mathcal{E},\mathbf{A}\}$, if $\mathcal{V}^{\prime}\subseteq\mathcal{V}$ and ${\mathcal{E}}^{\prime}= \left\{ {(u,v)|u,v \in {\mathcal{V}}^{\prime},(u,v) \in {\mathcal{E}}} \right\}$. 
\end{definition}

\begin{definition}
\cite{ref17}\cite{ref27} A node $v_i$ is defined as reachable from a node $v_j$, if there is a directed path from  $v_j$ to  $v_i$.  Similarly,  $v_i$ is defined as reachable from the set $\mathcal{W}$, if there exists a node $v_j\in \mathcal{W}$ such that $v_i$ is reachable from $v_j$. If there is a directed path from a leader to the node $v_i$, then  $v_i$ is defined as reachable from the leader set.
\end{definition}

\begin{definition}
\cite{ref27} A digraph is strongly connected if any two of its nodes are mutually reachable. A strongly connected component  is an induced subgraph that is maximal, and subject to being strongly connected. An independent strongly connected component is an induced subgraph such that it is a strongly connected component and there are no incoming edges from any other strongly connected components to any of its nodes. 
\end{definition}

\section{target controllability of a first-order multiagent system}\label{sec3}

Consider a first-order multiagent system with $n$ agents. Let $\mathcal{V}_L=\{v_1,\cdots,v_l\}\subseteq \mathcal{V}$, which is called the leader set. The nodes contained in $\mathcal{V}_L$ are called leaders, which are mainly manipulated by external inputs. $\mathcal{V}_F=\mathcal{V}\backslash \mathcal{V}_L$ is called the follower set, and $\mathcal{V}_T=\{\bar{v}_1,\cdots,\bar{v}_p\}\subseteq \mathcal{V}$ is called the target set. This paper mainly focuses on controlling the states of target nodes. The dynamics of each agent is given by

\begin{align}\label{1}
\begin{cases}
\dot{x}_i(t)=\sum\limits_{j \in {\mathcal{N}_i}}^{}{ a_{ij}}{[{x_j(t)} - {x_i(t)}]}+u_i(t), & i\in \mathcal{V}_L \\
\dot{x}_i(t)=\sum\limits_{j \in {\mathcal{N}_i}}^{}{ a_{ij}}{[{x_j(t)} - {x_i(t)}]},  & i\in  \mathcal{V}_F \\
y_i(t)=x_i(t), & i\in \mathcal{V}_T,
\end{cases}
\end{align}
where $x_i$ is the state of agent $i$, and $u_i$ is the control input. The compact form of dynamics \eqref{1} can be rewritten as

\begin{equation}
\begin{split}
\label{2}
\dot{x}(t)&=-Lx(t)+Bu(t),\\
y(t)&=Hx(t),
\end{split}
\end{equation}
where $x(t) = [x_1(t),\cdots,x_n(t)]^T$, $u(t) = [\tilde{u}_1(t), \cdots,\tilde{u}_l(t)]^T$, and $y(t) = [x_{\bar{v}_1}(t),\cdots,x_{\bar{v}_p}(t)]^T$ represent state vector, input vector, and output vector, respectively. $B = [e_{v_1},\cdots,e_{v_l}] \in \mathbb{R}^{n\times l}$, and $H = [e_{\bar{v}_1},\cdots,e_{\bar{v}_p}] ^T\in \mathbb{R}^{p\times n}$,  where $e_i$ represents the column where the $i$-th element is 1 and the others are 0. 

\begin{definition}
\cite{ref17}\cite{ref23}\cite{ref25}
System \eqref{2} is said to be target controllable, if there is a control input $u(t)$, so that any target terminal state $y(t_f)\in \mathbb{R}^p$ can be obtained from any initial state $x(t_0)\in \mathbb{R}^n$.
\end{definition}

\begin{lemma}\label{lem1}
\cite{ref17}\cite{ref23}
The target controllability is a special kind of the output controllability. System \eqref{2} is target controllable if and only if $W=[HB,H(-L)B,\cdots,H(-L)^{n-1}B]$ is  full row rank,
where  $W$ is called the target controllability matrix. Obviously, $W=HQ$, where $Q$ is the controllability matrix of the system.
\end{lemma}

\begin{definition}\cite{ref22}\cite{ref24}
The target controllable subspace is defined as the space of the target state $y(t)\in \mathbb{R}^p$ which is reached from the initial state $x(t_0)=0$ by a suitable control input $u(t)$. The dimension of the target controllable subspace of system \eqref{2} is expressed as $dim (-L,B,H)$,  and $dim(-L,B,H)=rank[HB,H(-L)B,\cdots,H(-L)^{n-1}B]$.
\end{definition}

\begin{definition}\label{def8}
 A matrix is called $S$ matrix, if it can be partitioned as 
$$
S= \begin{bmatrix}
 \begin{array}{ccc}
*&0&*\\ 
*&*&*\\ 
*&0&*
\end{array}
 \end{bmatrix},
$$
where $*$ represents zero block or nonzero block, $0$ is only of the same or different zero block, and the diagonal blocks take the form of square matrix.
\end{definition}

\begin{lemma}\label{lem3}
$S$ matrix has the following special properties.
\begin{enumerate}
\item $S\times S\doteq S$;
\item $S\pm S\doteq S$;
\item $S^n\doteq\underbrace{S\times S\times\cdots \times S}_{\text{n}} \doteq S$, where $n\in \mathbb{N}$;
\item $kS \doteq S$, where $k$ is any constant.
\end{enumerate}

Symbol  $\doteq$  only means that the matrices on both sides of $\doteq $ have the same form.
\end{lemma}
\begin{proof}
\begin{enumerate}
\item $S\times S=\begin{bmatrix}\setlength{\arraycolsep}{2.5pt}
 \begin{array}{ccc}
*&0&*\\ 
*&*&*\\ 
*&0&*
\end{array}
 \end{bmatrix}\begin{bmatrix}\setlength{\arraycolsep}{2.5pt}
 \begin{array}{ccc}
*&0&*\\ 
*&*&*\\ 
*&0&*
\end{array}
 \end{bmatrix}=\begin{bmatrix}\setlength{\arraycolsep}{2.5pt}
 \begin{array}{ccc}
*&0&*\\ 
*&*&*\\ 
*&0&*
\end{array}
 \end{bmatrix} \doteq S$;

\item $S\pm S=\begin{bmatrix}\setlength{\arraycolsep}{2.5pt}
 \begin{array}{ccc}
*&0&*\\ 
*&*&*\\ 
*&0&*
\end{array}
 \end{bmatrix}\pm\begin{bmatrix}\setlength{\arraycolsep}{2.5pt}
 \begin{array}{ccc}
*&0&*\\
*&*&*\\
*&0&*
\end{array}
 \end{bmatrix}=
\begin{bmatrix}\setlength{\arraycolsep}{2.5pt}
 \begin{array}{ccc}
*&0&*\\ 
*&*&*\\ 
*&0&*
\end{array}
 \end{bmatrix} \doteq S$; 

\item $S ^0=I\doteq S$, $S^1=S\doteq S$, and $S^2=S\times S\doteq S$. So, $S^n=\underbrace{S\times \cdots \times  S}_{\text{n}} \doteq\underbrace{S\times \cdots \times S}_{\text{n-1}}\doteq \cdots \doteq  S$;

\item $kS=k\begin{bmatrix}\setlength{\arraycolsep}{2.5pt}
 \begin{array}{ccc}
*&0&*\\ 
*&*&*\\ 
*&0&*
\end{array}
 \end{bmatrix} =\begin{bmatrix}\setlength{\arraycolsep}{2.5pt}
 \begin{array}{ccc}
*&0&*\\ 
*&*&*\\ 
*&0&*
\end{array}\end{bmatrix}\doteq S$, where $k$ is any constant.
\end{enumerate}
\end{proof}
\begin{definition}
 In a directed weighted topology $\mathcal{G}$, if $v_i$ is reachable from the leader set, which can be referred to simply as the reachable node. Conversely, $v_i$ can be called a unreachable node.  Reassigning the weight to 1, then there is a shortest path from a leader to a follower through  $\delta$ nodes.  We say that this follower is $\delta-$reachable from the leader set, shortened to the $\delta-$reachable node. In addition, all leaders are also defined as  reachable from the leader set.
\end{definition}

\begin{figure}[t]
\centerline{\includegraphics[width=210pt,height=150pt]{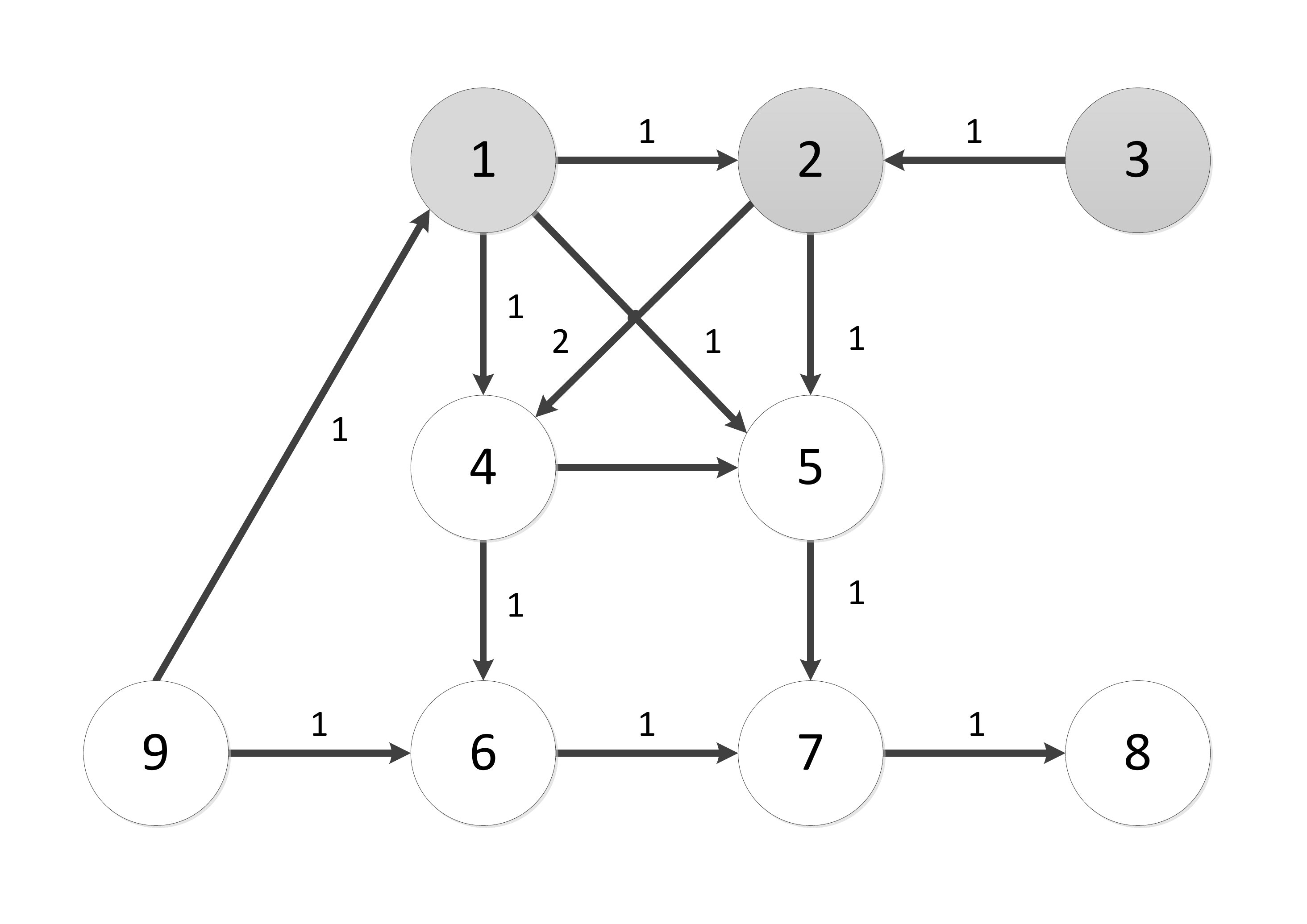}}
\caption{A directed weighted topology graph with nine nodes.\label{fig6}}
\end{figure}
 In Figure \ref{fig6}$, \mathcal{V}_L=\{1,2,3\}$ is  the leader set. Reassigning the weight to $1$, then there is a shortest path $1,4$ or $2,4$ from leaders $1$ or $2$ to the node $4$ through $0$ node. We say that this follower node $4$ is $0-$reachable from the leader set. Similarly,
node $5$ is $0-$reachable from the leader set, node $6$ is $1-$reachable from the leader set, node $7$ is $1-$reachable from the leader set, and node $8$ is $2-$reachable from the leader set. 

\begin{proposition}\label{pro1}
 If system \eqref{2} is target controllable, then each target node is reachable from the leader set. Meanwhile, that each target node is reachable from the leader set is equivalent to the fact that there are no zero rows in the target controllability matrix $W$.
\end{proposition}
\begin{proof}Assume that there are unreachable target nodes, and the unreachable target node set is $\mathcal{V}^{\prime}_{T} \subseteq{\mathcal{V}}_{T}$. WLOG, let the unreachable target node set be $\mathcal{V}^{\prime}_T=\{v_1,v_2,\ldots,v_s\}$, and the union of reachable target node set and leader set be $\mathcal{V}_C=\{v_{s+1},v_{s+2},\ldots,v_{s+c}\}$. The set of remaining nodes is ${\mathcal{V}_{D}}=\{v_{s+c+1},v_{s+c+2},...,v_{s+c+d}\}$.
Through rearranging the indices of agents in $\mathcal{G}$, the matrix ${L}$ is 
$$
  L= \bordermatrix{%
 &{{\mathcal{V}^{\prime}}_T}&{{\mathcal{V}}_{C}}&{{\mathcal{V}}_{D}}\cr
 {{\mathcal{V}^{\prime}}_T}&*&0_{s\times c}&*\cr
 {{\mathcal{V}}_{C}}&*&*&*\cr
 {{\mathcal{V}}_{D}}&*&0_{d\times c}&*
 }.
$$
 By Lemma \ref{lem3}, 
\[(-L)^ {n-1} = \left[ {\begin{array}{*{20}{c}}
 *&0_{s\times c}&*\\
 *&*&*\\
 *&0_{d\times c}&*
\end{array}} \right].\]
The output matrix ${H}$ can be expressed as
$$
 H=[e_1,e_2,\ldots,e_p]^T.
$$
Therefore, 
$$\begin{aligned}
W&=[HB,H(-L)B,\ldots,H(-L)^{n-1}B]\\&=H \left[ {\begin{array}{*{20}{c}}0_{s\times l}&0_{s
\times l}&\cdots&0_{s\times l}\\ \ast&\ast&\cdots&\ast\\0_{d\times l}&0_{d\times l}&\cdots&0_{d\times l}\end{array}} \right]\\&=
\left[ {\begin{array}{*{20}{c}}0_{s\times l}&0_{s\times l}&\cdots&0_{s\times l}\\ \star&\star&\cdots&\star\end{array}} \right],
\end{aligned}$$
where $\star$ is a zero or nonzero block. There are zero rows in the target controllability matrix $W$, so $rank(W)<p$. By Lemma \ref{lem1}, system \eqref{2} is not target controllable. This is a contradiction. Therefore, if the system \eqref{2} is target controllable, each target node is reachable from the leader set.

Meanwhile,
 $$\begin{aligned}
W=\left[ {\begin{array}{*{20}{c}}0_{s\times l}&0_{s\times l}&\cdots&0_{s\times l}\\ \star&\star&\cdots&\star\end{array}} \right].
\end{aligned}$$
So there are zero rows in the target controllability matrix $W$, which is a contradiction. Therefore, if there are no zero rows in the target controllability matrix $W$, each target node is reachable from the leader set.

Assume that each target node is reachable from the leader set, through rearranging the indices of agents in $\mathcal{G}$, then
$$
  L = \bordermatrix{%
 &{{\mathcal{V}_L}}&{{\mathcal{V}^{r}_L}}&{{\mathcal{V}}_{D}}\cr
 {{\mathcal{V}_L}}&-L_{11}&-L_{12}&\ast\cr
 {{\mathcal{V}^{r}_L}}&-L_{21}&-L_{22}&*\cr
 {{\mathcal{V}}_{D}}&0&0&*
 },
$$
where  $-L_{11}\in \mathbb{R}^{l\times l} $, $-L_{12}\in \mathbb{R}^{l\times r}$, $-L_{21}\in \mathbb{R}^{r\times l}$,  and $-L_{22}\in \mathbb{R}^{r\times r}$ are submatrixs of the matrix $L$.  $\mathcal{V}_L$ is  the leader set, ${\mathcal{V}^{r}_L}$ is  the reachable node set,  and $\mathcal{V}_D$ is the unreachable node set. 
The input matrix is
$$
 B=\left[ {\begin{array}{*{20}{c}}
 I_l\\
 0\\
 0
\end{array}} \right].
$$
It follows that the controllability matrix is
\begin{multline}
\begin{split}
 Q=&
\left[\begin{array}{ccc}
I_l&L_{11}&L_{11}^2+L_{12}L_{21}\\
0&L_{21}&L_{21}L_{11}+L_{22}L_{21}\\
0&0&0
\end{array}\right.\\
&\left.\begin{array}{cc}
L_{11}^3+L_{12}L_{21}L_{11}+L_{11}L_{12}^2+L_{12}L_{22}L_{12}&\ldots\\
L_{21}L_{11}^2+L_{22}L_{21}L_{11}+L_{21}L_{12}L_{21}+L_{22}^2L_{21}&\ldots\\
0&\ldots\\
\end{array}\right].\nonumber
\end{split}
\end{multline}
By elementary column transformation \cite{ref28} of block matrix $Q$, we get  that
\begin{multline}
\begin{split}
&\qquad \left[\begin{array}{ccc}
I_l&L_{11}&L_{11}^2+L_{12}L_{21}\\
0&L_{21}&L_{21}L_{11}+L_{22}L_{21}\\
0&0&0
\end{array}\right.\\
&\qquad \left.\begin{array}{cc}
L_{11}^3+L_{12}L_{21}L_{11}+L_{11}L_{12}^2+L_{12}L_{22}L_{12}&\ldots\\
L_{21}L_{11}^2+L_{22}L_{21}L_{11}+L_{21}L_{12}L_{21}+L_{22}^2L_{21}&\ldots\\
0&\ldots\\
\end{array}\right]\\&\longrightarrow \left[ {\begin{array}{*{20}{c}}
 I_l&0&0&\ldots&0\\
 0&L_{21}&L_{22}L_{21}&\ldots&L_{22}^{n-2}L_{21}\\
 0&0&0&\ldots&0\\
\end{array}} \right]\\&\,\,\,\,=\,\, Q_{et}.\nonumber
\end{split}
\end{multline}
 Next, we prove that the rows $1$ to $l+r$ of $Q_{et}$ are nonzero rows.
 Suppose $\delta=0,1,\ldots,\tau$ for $\tau\in \mathbb{N}$, then reachable nodes can be divided into $\tau+1$ classes according to the size of $\delta$. Suppose there are $\tilde{r}_\delta$ $\delta-$reachable nodes.                                                                                                                                                                                                                                                                                                                                                                                                                                                                                                                                                                                                                                                             
Let
$$\begin{aligned}
r_\delta=
\begin{cases}
\tilde{r}_{\delta}, &\delta=0\cr
r_{\delta-1}+\tilde{r}_{\delta}, &\delta=1,2,\ldots,\tau,\cr \end{cases}
\end{aligned}
$$
where $r_\tau=r$, and  $r$ is the number of reachable nodes. The following formula is obtained by the permutation of nodes,
\begin{strip}
 \hrulefill
\begin{align}
  L_{22}= \setlength{\arraycolsep}{0.01pt}{\addtocounter{MaxMatrixCols}{32}
\begin{bmatrix}
 \beta^{(1)}_{11}&\ldots&\beta^{(1)}_{1r_0}&\beta^{(1)}_{1(r_0+1)}&\ldots&\beta^{(1)}_{1r_1}&\beta^{(1)}_{1(r_1+1)}&\ldots&\beta^{(1)}_{1r_2}&\ldots&\beta^{(1)}_{1(r_{\tau-1}+1)}&\ldots&\beta^{(1)}_{1r}\\
\vdots&\ddots &\vdots&\vdots&\qquad &\vdots&\vdots&\qquad &\vdots&\qquad &\vdots&\qquad&\vdots\\
\beta^{(1)}_{r_01}&\ldots&\beta^{(1)}_{r_0r_0}&\beta^{(1)}_{r_0(r_0+1)}&\ldots&\beta^{(1)}_{r_0r_1}&\beta^{(1)}_{r_0(r_1+1)}&\ldots&\beta^{(1)}_{r_0r_2}&\ldots&\beta^{(1)}_{r_0(r_{\tau-1}+1)}&\ldots&\beta^{(1)}_{r_0r}\\
\beta^{(1)}_{(r_0+1)1}&\ldots&\beta^{(1)}_{(r_0+1)r_0}&\beta^{(1)}_{(r_0+1)(r_0+1)}&\ldots&\beta^{(1)}_{(r_0+1)r_1}&\beta^{(1)}_{(r_0+1)(r_1+1)}&\ldots&\beta^{(1)}_{(r_0+1)r_2}&\ldots&\beta^{(1)}_{(r_0+1)(r_{\tau-1}+1)}&\ldots&\beta^{(1)}_{(r_0+1)r}\\
\vdots&\qquad &\vdots&\vdots&\ddots &\vdots&\vdots&\qquad &\vdots&\qquad &\vdots&\qquad&\vdots\\
\beta^{(1)}_{r_11}&\ldots&\beta^{(1)}_{r_1r_0}&\beta^{(1)}_{r_1(r_0+1)}&\ldots&\beta^{(1)}_{r_1r_1}&\beta^{(1)}_{r_1(r_1+1)}&\ldots&\beta^{(1)}_{r_1r_2}&\ldots&\beta^{(1)}_{r_1(r_{\tau-1}+1)}&\ldots&\beta^{(1)}_{r_1r}\\
0&\ldots&0&\beta^{(1)}_{(r_1+1)(r_0+1)}&\ldots&\beta^{(1)}_{(r_1+1)r_1}&\beta^{(1)}_{(r_1+1)(r_1+1)}&\ldots&\beta^{(1)}_{(r_1+1)r_2}&\ldots&\beta^{(1)}_{(r_1+1)(r_{\tau-1}+1)}&\ldots&\beta^{(1)}_{(r_1+1)r}\\
\vdots&\qquad &\vdots&\vdots&\qquad &\vdots&\vdots&\ddots &\vdots&\qquad &\vdots&\qquad&\vdots\\
0&\ldots&0&\beta^{(1)}_{r_2(r_0+1)}&\ldots&\beta^{(1)}_{r_2r_1}&\beta^{(1)}_{r_2(r_1+1)}&\ldots&\beta^{(1)}_{r_2r_2}&\ldots&\beta^{(1)}_{r_2(r_{\tau-1}+1)}&\ldots&\beta^{(1)}_{r_2r}\\
\vdots&\qquad &\vdots&\vdots&\qquad &\vdots&\vdots&\qquad &\vdots&\ddots&\vdots&\qquad&\vdots\\
0&\ldots&0&0&\ldots&0&0&\ldots&0&\ldots&\beta^{(1)}_{(r_{\tau-1}+1)(r_{\tau-1}+1)}&\ldots&\beta^{(1)}_{(r_{\tau-1}+1)r}\\
\vdots&\qquad &\vdots&\vdots&\qquad &\vdots&\vdots&\qquad &\vdots&\qquad &\vdots&\ddots&\vdots\\
0&\ldots&0&0&\ldots&0&0&\ldots&0&\ldots&\beta^{(1)}_{r(r_{\tau-1}+1)}&\ldots&\beta^{(1)}_{rr}\\
\end{bmatrix}},\nonumber
\end{align}
\hrulefill
\end{strip}
 $$L_{21}=
\begin{bmatrix}
 \iota_{(l+1)(l+1)}&\ldots&\iota_{(l+r_0)(l+1)}& 0& \ldots&0\\
\vdots&\qquad &\vdots&\vdots&\qquad&\vdots\\
\iota_{(l+1)(2l)}&\ldots&\iota_{(l+r_0)(2l)}&0&\ldots&0\\
\end{bmatrix}^T,$$
where $\iota_{\tilde{\epsilon}_0(l+1)},\ldots,\iota_{\tilde{\epsilon}_0(2l)}$ are all nonnegative and at least one value of $\iota_{\tilde{\epsilon}_0(l+1)},\ldots,\iota_{\tilde{\epsilon}_0(2l)}$ is  positive, $\tilde{\epsilon}_0=l+1,\ldots,l+r_0$; and $\beta^{(1)}_{\epsilon_{\bar{\delta}}(r_{\bar{\delta}-2}+1)},\ldots,\beta^{(1)}_{\epsilon_{\bar{\delta}} r_{\bar{\delta}-1}}$ are all nonnegative and at least one value of $\beta^{(1)}_{\epsilon_{\bar{\delta}}(r_{\bar{\delta}-2}+1)},\ldots,\beta^{(1)}_{\epsilon_{\bar{\delta}} r_{\bar{\delta}-1}}$ is positive, $\epsilon_{\bar{\delta}}=r_{\bar{\delta}-1}+1,\ldots,r_{\bar{\delta}}$, ${\bar{\delta}}=1,2,\ldots,\tau$, and $r_{-1}=0$.
Let $$L_{22}^\delta=\left[\begin{array}{*{20}{c}}
\beta_{11}^{(\delta)}&\ldots& \beta_{1r}^{(\delta)}\\
 \vdots&\ddots&\vdots\\
\beta_{r1}^{(\delta)}&\ldots&\beta_{rr}^{(\delta)}\\
\end{array}\right],$$$$
L_{22}^\delta L_{21}=\setlength{\arraycolsep}{1.5pt} {\addtocounter{MaxMatrixCols}{32}
\begin{bmatrix}
\iota_{(l+1)(\delta l+l+1)}&\ldots& \iota_{(l+1)(\delta l+2l)}\\
 \vdots&\ddots&\vdots\\
\iota_{(l+r)(\delta l+l+1)}&\ldots&\iota_{(l+r)(\delta l+2l)}\\
\end{bmatrix}}.$$
There is a block $I_l$, which consists of elements at the intersection of the first $l$ rows and the first $l$ columns of the matrix $Q_{et}$. So rows $1$ to $l$ of matrix $Q_{et}$ are nonzero rows.
\begin{enumerate}
\item  $\iota_{\tilde{\epsilon}_0(l+1)},\ldots,\iota_{\tilde{\epsilon}_0(2l)}$ are all nonnegative and at least one value of $\iota_{\tilde{\epsilon}_0(l+1)},\ldots,\iota_{\tilde{\epsilon}_0(2l)}$ is  positive. Hence the rows  $l+1$ to $l+r_0$ of matrix $Q_{et}$ are nonzero rows. 

\item $\beta^{(1)}_{\epsilon_11},\ldots,\beta^{(1)}_{\epsilon_1 r_0}$ are all nonnegative and at least one value of $\beta^{(1)}_{\epsilon_11},\ldots,\beta^{(1)}_{\epsilon_1 r_0}$ is positive. Suppose $\beta^{(1)}_{\epsilon_1j_1}$ is positive for $j_1\in\{1,\ldots,r_0\}$, then at least one value in $\beta^{(1)}_{\epsilon_1j_1}\iota_{\tilde{\epsilon}_0(l+1)}$,$\beta^{(1)}_{\epsilon_1j_1}\iota_{\tilde{\epsilon}_0(l+2)}$,$\ldots$,$\beta^{(1)}_{\epsilon_1j_1}\iota_{\tilde{\epsilon}_0(2l)}$  is positive. By calculating $L_{22}L_{21}$,
$\iota_{(l+r_0+1)(2l+1)}=\beta^{(1)}_{(r_0+1)1}\iota_{(l+1)(l+1)}+\ldots+\beta^{(1)}_{(r_0+1)r_0}\iota_{(l+r_0)(l+1)}$, 
$\iota_{(l+r_0+1)(2l+2)}=\beta^{(1)}_{(r_0+1)1}\iota_{(l+1)(l+2)}+\ldots+\beta^{(1)}_{(r_0+1)r_0}\iota_{(l+r_0)(l+2)}$, 
$\ldots$, 
$\iota_{(l+r_0+1)(3l)}=\beta^{(1)}_{(r_0+1)1}\iota_{(l+1)(2l)}+\ldots+\beta^{(1)}_{(r_0+1)r_0}\iota_{(l+r_0)(2l)}$. 
$\beta^{(1)}_{(r_0+1)1}$, $\beta^{(1)}_{(r_0+1)2}$, $\ldots$, $\beta^{(1)}_{(r_0+1)r_0}$, $\iota_{\tilde{\epsilon}_0(l+1)}$, $\ldots$, $\iota_{\tilde{\epsilon}_0(2l)}$ are nonnegative. So $\beta^{(1)}_{(r_0+1)1}\iota_{(l+1)(l+1)}$, $\ldots$, $\beta^{(1)}_{(r_0+1)r_0}\iota_{(l+r_0)(l+1)}$, $\beta^{(1)}_{(r_0+1)1}\iota_{(l+1)(l+2)}$, $\ldots$, $\beta^{(1)}_{(r_0+1)r_0}\iota_{(l+r_0)(l+2)}$, $\ldots$, $\beta^{(1)}_{(r_0+1)1}\iota_{(l+1)(2l)}$, $\ldots$, $\beta^{(1)}_{(r_0+1)r_0}\iota_{(l+r_0)(2l)}$ are nonnegative. There are $\beta^{(1)}_{\epsilon_1j_1}\iota_{(l+j_1)(l+1)}$, $\beta^{(1)}_{\epsilon_1j_1}\iota_{(l+j_1)(l+2)}$, $\ldots$, $\beta^{(1)}_{\epsilon_1j_1}\iota_{(l+j_1)(2l)}$ in $\beta^{(1)}_{(r_0+1)1}\iota_{(l+1)(l+1)}$, $\ldots$, $\beta^{(1)}_{(r_0+1)r_0}\iota_{(l+r_0)(l+1)}$, $\beta^{(1)}_{(r_0+1)1}\iota_{(l+1)(l+2)}$, $\ldots$, $\beta^{(1)}_{(r_0+1)r_0}\iota_{(l+r_0)(l+2)}$, $\ldots$, $\beta^{(1)}_{(r_0+1)1}\iota_{(l+1)(2l)}$, $\ldots$, $\beta^{(1)}_{(r_0+1)r_0}\iota_{(l+r_0)(2l)}$, and at least one value of $\beta^{(1)}_{\epsilon_1j_1}\iota_{(l+j_1)(l+1)}$, $\beta^{(1)}_{\epsilon_1j_1}\iota_{(l+j_1)(l+2)}$, $\ldots$, $\beta^{(1)}_{\epsilon_1j_1}\iota_{(l+j_1)(2l)}$ is positive. Therefore, at least one value of $\beta^{(1)}_{(r_0+1)1}\iota_{(l+1)(l+1)}$, $\ldots$, $\beta^{(1)}_{(r_0+1)r_0}\iota_{(l+r_0)(l+1)}$, $\beta^{(1)}_{(r_0+1)1}\iota_{(l+1)(l+2)}$, $\ldots$, $\beta^{(1)}_{(r_0+1)r_0}\iota_{(l+r_0)(l+2)}$, $\ldots$, $\beta^{(1)}_{(r_0+1)1}\iota_{(l+1)(2l)}$, $\ldots$, $\beta^{(1)}_{(r_0+1)r_0}\iota_{(l+r_0)(2l)}$ is positive. Since  $\beta^{(1)}_{(r_0+1)1}\iota_{(l+1)(l+1)}$, $\ldots$, $\beta^{(1)}_{(r_0+1)r_0}\iota_{(l+r_0)(l+1)}$, $\beta^{(1)}_{(r_0+1)1}\iota_{(l+1)(l+2)}$, $\ldots$, $\beta^{(1)}_{(r_0+1)r_0}\iota_{(l+r_0)(l+2)}$, $\ldots$, $\beta^{(1)}_{(r_0+1)1}\iota_{(l+1)(2l)}$, $\ldots,\beta^{(1)}_{(r_0+1)r_0}\iota_{(l+r_0)(2l)}$ is  nonnegative, then at least one value of $\iota_{(l+r_0+1)(2l+1)}$, $\ldots$, $\iota_{(l+r_0+1)(3l)}$ is  positive. Therefore, the $(l+r_0+1)$-th row of matrix $Q_{et}$ is nonzero row. Similarly, the rows $l+r_0+2$ to $l+r_1$ of matrix $Q_{et}$ are all nonzero rows.

\item $\beta^{(1)}_{\epsilon_21},\ldots,\beta^{(1)}_{\epsilon_2 r_0}$  are all nonnegative and at least one value of $\beta^{(1)}_{\epsilon_21},\ldots,\beta^{(1)}_{\epsilon_2 r_0}$  is positive. Suppose $\beta^{(1)}_{\epsilon_2j_2}$ is positive for $j_2\in\{r_0+1,\ldots,r_1\}$,  then  at least one value in $\beta^{(1)}_{\epsilon_2j_2}\beta^{(1)}_{\epsilon_11}$,$\beta^{(1)}_{\epsilon_2j_2}\beta^{(1)}_{\epsilon_12}$,$\ldots$,$\beta^{(1)}_{\epsilon_2j_2}\beta^{(1)}_{\epsilon_1r_0}$ is positive.
By calculating $L^2_{22}$, 
$\beta^{(2)}_{(r_1+1)1}=\beta^{(1)}_{(r_1+1)(r_0+1)}\beta^{(1)}_{(r_0+1)1}+\ldots+\beta^{(1)}_{(r_1+1)r_1}\beta^{(1)}_{r_11}$, 
$\beta^{(2)}_{(r_1+1)2}=\beta^{(1)}_{(r_1+1)(r_0+1)}\beta^{(1)}_{(r_0+1)2}+\ldots+\beta^{(1)}_{(r_1+1)r_1}\beta^{(1)}_{r_12}$, 
$\ldots$, 
$\beta^{(2)}_{(r_1+1)(r_0)}=\beta^{(1)}_{(r_1+1)(r_0+1)}\beta^{(1)}_{(r_0+1)r_0}+\ldots+\beta^{(1)}_{(r_1+1)r_1}\beta^{(1)}_{r_1r_0}$. 
Similar to Step 2, it can be seen that at least one value of $\beta_{\epsilon_21}^{(2)},\beta_{\epsilon_22}^{(2)},\ldots,\beta_{\epsilon_2r_0}^{(2)}$ is positive. Suppose $\beta_{\epsilon_2j_1}^{(2)}$ is positive for $j_1\in\{1,\ldots,r_0\}$, then at least one value in $\beta_{\epsilon_2j_1}^{(2)}\iota_{\tilde{\epsilon}_0(l+1)},\beta_{\epsilon_2j_1}^{(2)}\iota_{\tilde{\epsilon}_0(l+2)},\ldots,\beta_{\epsilon_2j_1}^{(2)}\iota_{\tilde{\epsilon}_0(2l)}$ is positive. Next, we consider $L_{22}^{2}L_{21}$. Similar to Step 2, it can  be  seen  that the rows $l+r_1+1$ to $l+r_2$ of the matrix $Q_{et}$ are all nonzero rows.
\end{enumerate}
The rest can be done in the same manner. In Step $\tau$, the rows  $l+r_{\tau-1}+1$ to $l+r_\tau$ of matrix $Q_{et}$ are all nonzero, where $r_\tau=r$.

In summary, the rows $1$ to $l+r$ of the matrix $Q_{et}$ are all nonzero rows. Elementary column transformation will not convert nonzero rows into zero rows. Therefore, if rows $1$ to $l+r$ of $Q_{et}$ are nonzero rows, so are rows $1$ to $l+r$ of $Q$.
 Let $Q=[q_1,\ldots,q_l,q_{l+1},\ldots,q_{l+r},q_{l+r+1},\ldots,q_{l+r+d}]^T$,
 where $q_i^T$  is  the $i$-th  row  of  the controllability  matrix $Q$, $i=1,2,\ldots,l+r+d$, $l+r+d=n$.  Since rows $1$ to $l+r$ of $Q$ are nonzero rows, $q_\mu^T\neq0$ for $\mu=1,2,\ldots,l+r$. All target nodes are reachable from the leader set,  so $W=[q_{i_1},q_{i_2},\ldots,q_{i_p}]^T$ for  $\{q^T_{i_1},q^T_{i_2},\ldots,q^T_{i_p}\}\subseteq \{q^T_{1},q^T_{2},\ldots,q^T_{l+r}\}$. Therefore, there are no zero rows in $W$. 
\end{proof}

\begin{remark}
Different from the existing results, which only consider part of the situation where there are unreachable nodes. The proof of  Proposition \ref{pro1} applies to all situations. According to the proof of Proposition \ref{pro1}, if there are $s$ unreachable target nodes, then 
$\begin{aligned}
  dim(-L,B,H)=rankW=rank
  \left[ \setlength{\arraycolsep}{1.3pt}{\begin{array}{*{20}{c}}0_{s\times l}&\cdots&0_{s\times l}\\ \star&\cdots&\star\end{array}} \right]
  \leq p-s
\end{aligned}$. Namely, the dimension of the target controllable subspace satisfies $dim(-L,B,H)\leq p-s$. It's important to note that reachable target node is a necessary condition for target controllability, not necessarily a sufficient one. A counter example is as follows.
\begin{figure}[t]
\centerline{\includegraphics[width=100pt,height=60pt]{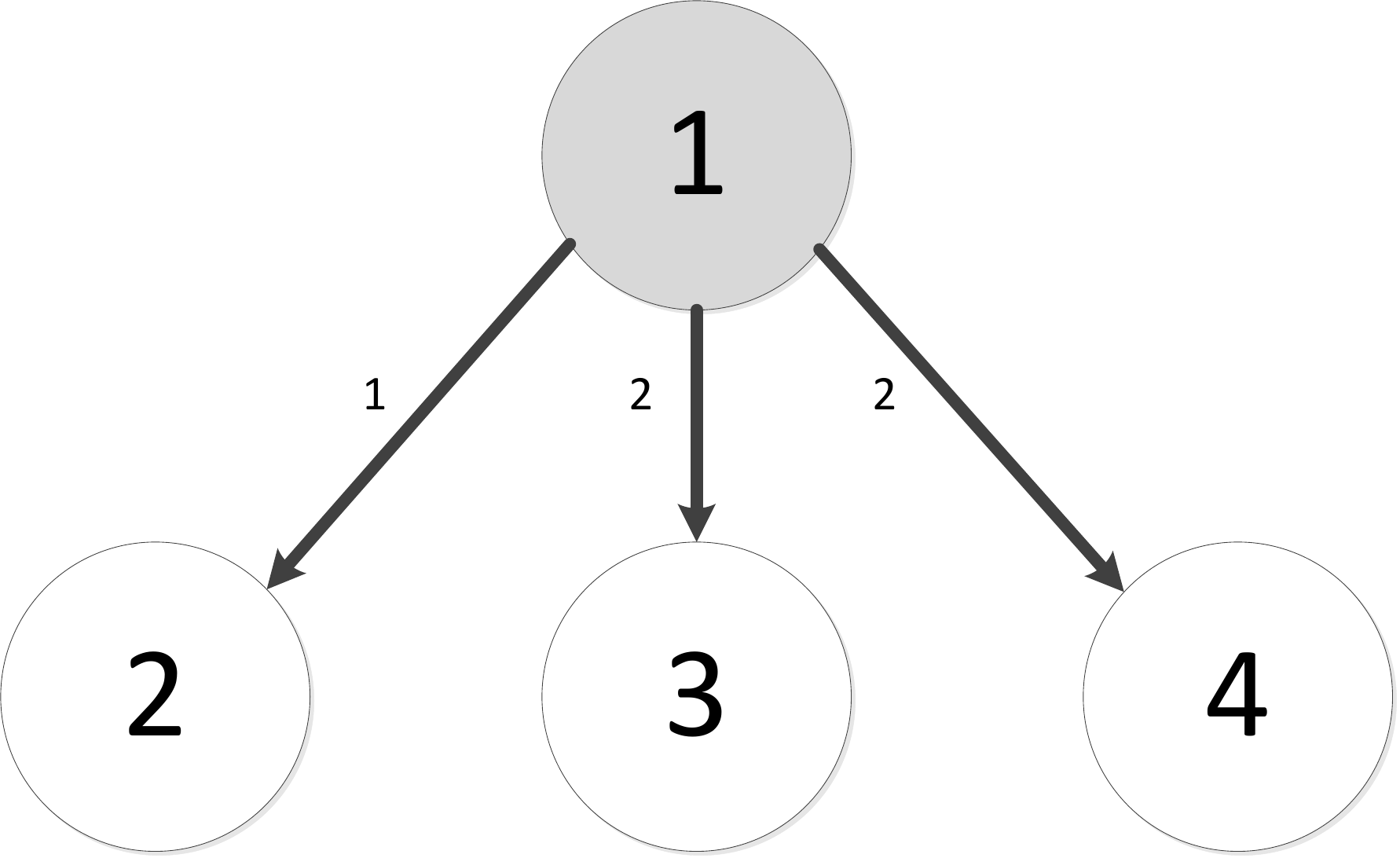}}
\caption{A counter example diagram (1).\label{fig11}}
\end{figure}
In Figure \ref{fig11}, $\mathcal{V}_L=\{1\}$ is the leader node set, other nodes are follower nodes, and $\mathcal{V}_T=\{3,4\}$ is the target node set. Obviously, target nodes are reachable from the leader set. If Proposition \ref{pro1} is a  sufficient condition, then system \eqref{2} is target controllable. 
We see that
\[(-L)= \left[\setlength{\arraycolsep}{2pt} {\begin{array}{*{20}{c}}
0&0&0&0\\
1&-1&0&0\\
2&0&-2&0\\
2&0&0&-2\\
\end{array}} \right],
B= \left[\setlength{\arraycolsep}{2pt} {\begin{array}{*{20}{c}}
1\\
0\\
0\\
0
\end{array}} \right],
H= \left[\setlength{\arraycolsep}{2pt} {\begin{array}{*{20}{c}}
0&0&1&0\\
0&0&0&1
\end{array}} \right].\]

$$\begin{aligned}rankW&=rank[HB,H(-L)B,H(-L)^2B,H(-L)^3B]\\&=rank\left[ {\begin{array}{*{20}{c}}
0&2&-4&8\\
0&2&-4&8
\end{array}} \right]\\&=1<2.\end{aligned}$$
By Lemma \ref{lem1}, the system is not target controllable, which contradicts the target controllability of system \eqref{2}. Therefore, Proposition \ref{pro1} is not a sufficient condition.
\end{remark}
\begin{definition}
A partition of graphs is defined as $\pi_{0{EP}}$, where $\pi_{0{EP}}=\{\{v_1\},\ldots,\{v_l\},\tilde{\mathcal{C}}_1,\ldots,\tilde{\mathcal{C}}_\zeta\}$, with  $v_1,\ldots,v_l$ taking leaders' role.  $\mathcal{C}_{l+1}=\{v_{l+1},\ldots,v_{n}\}$ which is divided by the maximum equitable partition $\pi_{EP}=\{\tilde{\mathcal{C}}_1,\ldots,\tilde{\mathcal{C}}_\zeta\}$ with respect to $\{v_1\},\ldots,\{v_l\}$.
\end{definition}

\begin{theorem}\label{theo1}
System \eqref{2} is target controllable if and only if each cell contains no more than one target node and there are no unreachable target nodes, with $\delta-$reachable nodes belonging to the same cell in $\pi_{0{EP}}$.
\end{theorem}
\begin{proof}According to the proof of Proposition \ref{pro1} and calculation, one has\\
\begin{strip}
 \hrulefill
\begin{align}
Q_{et}= \setlength{\arraycolsep}{0.01pt}{\addtocounter{MaxMatrixCols}{32}
\begin{bmatrix}
 1&\ldots&0&0&\ldots&0&\ldots&0&\ldots&0&0&\ldots&0\\
\vdots&\ddots &\vdots&\vdots&\qquad &\vdots&\qquad &\vdots&\qquad&\vdots\\
0&\ldots&1&0&\ldots&0&\ldots&0&\ldots&0&0&\ldots&0\\
0&\ldots&0&\iota_{(l+1)(l+1)}&\ldots&\iota_{(l+1)(2l)}&\ldots&\divideontimes&\ldots&\divideontimes&\divideontimes&\ldots&\divideontimes\\
\vdots&\qquad &\vdots&\vdots&\ddots &\vdots&\qquad &\vdots&\qquad&\vdots&\vdots&\qquad&\vdots\\
0&\ldots&0&\iota_{(l+r_0)(l+1)}&\ldots&\iota_{(l+r_0)(2l)}&\ldots&\divideontimes&\ldots&\divideontimes&\divideontimes&\ldots&\divideontimes\\
\vdots&\qquad &\vdots&\vdots&\qquad &\vdots&\ddots&\vdots&\qquad&\vdots&\vdots&\qquad&\vdots\\
0&\ldots&0&0&\ldots&0&\ldots&\iota_{(l+r_{\tau-1}+1)(\tau l+l+1)}&\ldots&\iota_{(l+r_{\tau-1}+1)(\tau l+2l)}&\divideontimes&\ldots&\divideontimes\\
\vdots&\qquad &\vdots&\vdots&\qquad &\vdots&\qquad &\vdots&\ddots&\vdots&\vdots&\qquad&\vdots\\
0&\ldots&0&0&\ldots&0&\ldots&\iota_{(l+r_{\tau})(\tau l+l+1)}&\ldots&\iota_{(l+r_{\tau})(\tau l+2l)}&\divideontimes&\ldots&\divideontimes\\
0&\ldots&0&0&\ldots&0&\ldots&0&\ldots&0&0&\ldots&0\\
\vdots&\qquad &\vdots&\vdots&\qquad &\vdots&\qquad &\vdots&\qquad&\vdots&\vdots&\qquad&\vdots\\
0&\ldots&0&0&\ldots&0&\ldots&0&\ldots&0&0&\ldots&0\\
\end{bmatrix}}.\nonumber
\end{align}
\hrulefill
\end{strip}
\noindent Let $Q_{et}=[\tilde{q}_1,\ldots,\tilde{q}_l,\tilde{q}_{l+1},\ldots,\tilde{q}_{l+r_0},\tilde{q}_{l+r_0+1},\ldots,\tilde{q}_{l+r_1},$\\$\ldots,\tilde{q}_{l+r_{\tau-1}+1},\ldots, \tilde{q}_{l+r_\tau},\tilde{q}_{l+r_\tau+1},\ldots,\tilde{q}_{l+r_\tau+d}]^T$, where $r_\tau=r$. 
Set
$$\begin{aligned}
\tilde{\epsilon}_\delta=
\begin{cases}
l+1,\ldots,l+r_0, &\delta=0\cr
l+r_{\delta-1}+1,\ldots,l+r_\delta, &\delta=1,\ldots,\tau.\cr \end{cases}
\end{aligned}
$$
We see that rows $\tilde{q}_1,\ldots,\tilde{q}_l,\tilde{q}_{\tilde{\epsilon}_0},\tilde{q}_{\tilde{\epsilon}_1},\ldots,q_{\tilde{\epsilon}_\tau}$ of matrix $Q_{et}$ are linearly independent. So the rows $q_1,\ldots,q_l,q_{\tilde{\epsilon}_0},q_{\tilde{\epsilon}_1},\ldots,q_{\tilde{\epsilon}_\tau}$ of the matrix $Q$ are linearly independent.

Suppose nodes $i$ and $j$ are in the same cell, with $\delta-$reachable nodes belonging to the same cell in $\pi_{0{EP}}$. There is a permutation matrix $J$ which can exchange the $(i-l)$-th and $(j-l)$-th rows of matrix by left multipying the matrix $J$, and $J$ satisfies $JL_{22}J^T=L_{22}$, $JL_{21}=L_{21}$. Hence
$$\begin{aligned}&\quad\,\,\,J[L_{21},L_{22}L_{21},L_{22}^2L_{21},\ldots,L_{22}^ {n-2}L_{21}]\\&=
[JL_{21},JL_{22}J^TJL_{21},\ldots,(JL_{22}J^{T})^{ n-2}JL_{21}]\\&=
[L_{21},L_{22}L_{21},L_{22}^2L_{21},\ldots,L_{22}^{n-2}L_{21}].\end{aligned}$$
Namely, the $(i-l)$-th and $(j-l)$-th rows of $[L_{21},L_{22}L_{21},L_{22}^2L_{21},\ldots,L_{22}^{n-2}L_{21}]$ are the same.
$$Q_{et}=\left[ {\begin{array}{*{20}{c}}
 I_l&0&0&\ldots&0\\
 0&L_{21}&L_{22}L_{21}&\ldots&L_{22}^{n-2}L_{21}\\
 0&0&0&\ldots&0\\
\end{array}} \right].$$
Consequently,  $\tilde{q}_i=\tilde{q}_j$. Then $q_i=q_j$. If nodes $i$ and $j$ are both $\delta-$reachable nodes, then $\tilde{q}_i=\tilde{q}_j\neq0$, $q_i=q_j\neq0$.

(Necessity). Assume that there is one cell containing no less than two target nodes, with $\delta-$reachable nodes belonging to the same cell in $\pi_{0{EP}}$.  Namely, target nodes $\bar{v}_i$ and $\bar{v}_j$ are in the same cell and $\delta-$reachable nodes belong to the same cell in $\pi_{0{EP}}$. Thus, $q_{\bar{v}_i}=q_{\bar{v}_j}$. Then, $rankW=rank(HQ)<p$, and accordingly system \eqref{2} is not target controllable. Assume that there are unreachable target nodes, by Proposition \ref{pro1}, system \eqref{2} is not target controllable.

(Sufficiency).  Assume that each cell contains no more than one target node and there are no unreachable target nodes, with $\delta-$reachable nodes belonging to the same cell in $\pi_{0{EP}}$. Since the rows $q_1,\ldots,q_l,q_{\tilde{\epsilon}_0},q_{\tilde{\epsilon}_1},\ldots,q_{\tilde{\epsilon}_\tau}$ of matrix $Q$  are linearly independent,  $rankW=rank(HQ)=p$. Thus, system \eqref{2} is target controllable.
\end{proof}

\begin{remark}\label{rem3}

\begin{figure}[t]
\centerline{\includegraphics[width=110pt,height=150pt]{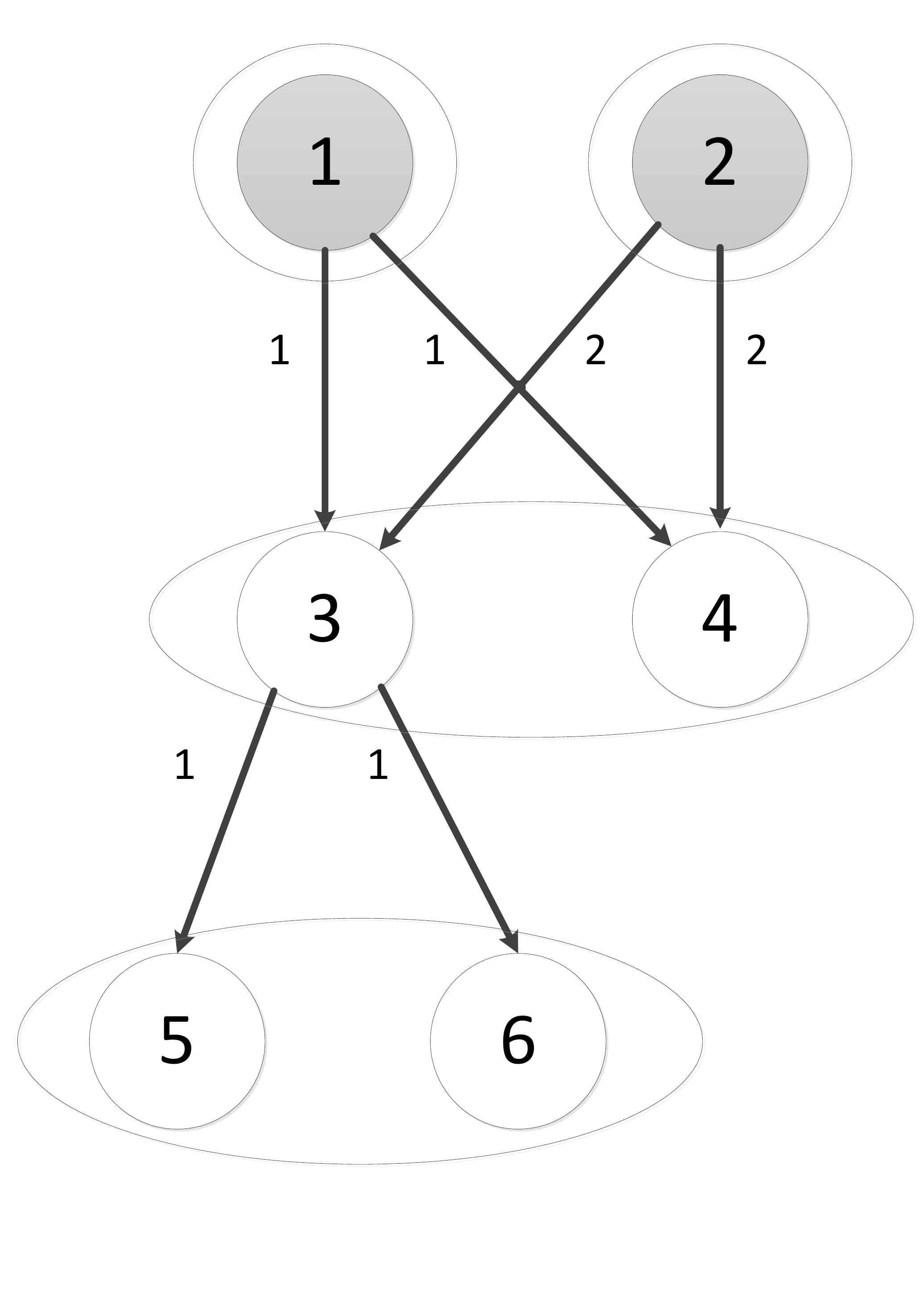}}
\caption{A directed weighted topology graph with six nodes.\label{fig10}}
\end{figure}

In Figure \ref{fig10},  a partition $\pi_{0{EP}}=\{\{1\},\{2\},\{3,4\},\{5,6\}\}$ is constructed. Nodes $1$ and $2$ are divided into different cells, and are selected as leaders. Let  $\mathcal{C}_{3}=\{3,4,5,6\}$, and the nodes of $\mathcal{C}_{3}$ be divided by taking advantage of the maximum equitable partition $\pi_{EP}=\{\tilde{\mathcal{C}}_1,\tilde{\mathcal{C}}_2\}$ with respect to $\{1\},\{2\}$.  Then,  $\tilde{\mathcal{C}}_1=\{3,4\}$, $\tilde{\mathcal{C}}_2=\{5,6\}$. Finally, we get the partition $\pi_{0{EP}}=\{\{1\},\{2\},\{3,4\},\{5,6\}\}$. Nodes $3$ and $4$ are $0-$reachable nodes, and nodes $5$ and $6$ are $1-$reachable nodes. Figure \ref{fig10} shows that $\delta-$reachable nodes belong to the same cell in $\pi_{0{EP}}$.
\end{remark}

\begin{corollary}\label{cor1}
If target nodes are chosen from $v_1$, $\ldots$, $v_l$, $v_{\tilde{\epsilon}_0}$, $\ldots$, $v_{\tilde{\epsilon}_\tau}$, then system \eqref{2} is target controllable, where $v_1$, $\ldots$, $v_l$ are leaders, $v_{\tilde{\epsilon}_{\delta}}$ is a $\delta-$reachable node, and $\delta=0,1,2,\ldots,\tau$.
\end{corollary}
\begin{proof}
It can be derived directly from the proof of Theorem \ref{theo1}.
\end{proof}

\begin{proposition}\label{pro2}
If system \eqref{2} is target controllable, then the following two statements are true.
\begin{enumerate}
\item $rankH[\lambda_iI-(-L),B]=p$,
where $\lambda_i$  is an eigenvalue of $-L$,  $i=1,2,\ldots,n$, and $p$ is the number of target nodes. 
\item there is no vector $q\in \mathbb{R}^p$ such that the nonzero left eigenvector $\vartheta^T$ of $-L$ satisfies $\vartheta^T=q^TH$ and  $\vartheta^TB=0$.
\end{enumerate}
\end{proposition}
\begin{proof}
\begin{enumerate}
\item   Assume that  $rankH[\lambda_iI-(-L),B]<p$. Then there is a nonzero vector $q$ such that $q^TH[\lambda_iI-(-L),B]=0$. 
As a consequence, $q^TH(-L)=\lambda_iq^TH$, $q^THB=0$. Therefore, $q^THB=0,\quad q^TH(-L)B=\lambda_iq^THB=0$, $\ldots$, $q^TH(-L)^{n-1}B=0$. 
Thus, $q^T[HB,H(-L)B,\cdots,H(-L)^{n-1}B]=q^TW=0$. Since $q$ is a nonzero vector, $rankW<p$. By Lemma \ref{lem1}, the system is not target controllable, which is a contradiction.
\item Suppose that there is a $q\in \mathbb{R}^p$, such that the nonzero left eigenvector $\vartheta^T$ of $-L$ satisfies $\vartheta^T=q^TH$ and  $\vartheta^TB=0$. Then, following the same arguments as the proof of the first part, one has  $rankW<p$.  By Lemma \ref{lem1}, the system is not target controllable, which contradicts the target controllability of system \eqref{2}.
\end{enumerate}
\end{proof}

Proposition \ref{pro2} is only a necessary condition for target controllability, which is illustrated by the following example.

\begin{figure}[t]
\centerline{\includegraphics[width=150pt,height=25pt]{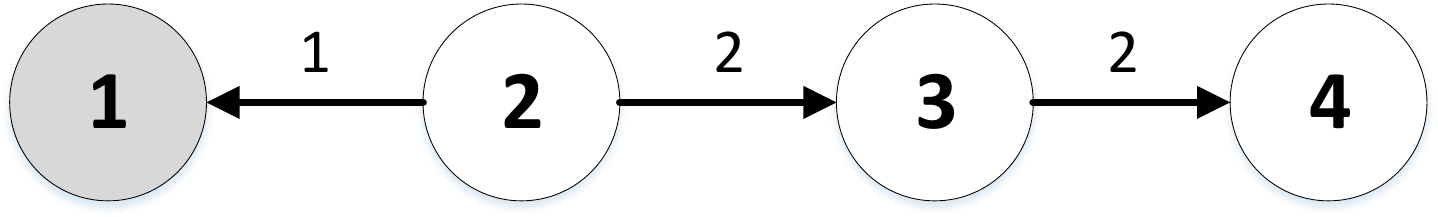}}
\caption{A counter example diagram (2).\label{fig7}}
\end{figure}

In Figure \ref{fig7}, $\mathcal{V}_L=\{1\}$ is the leader node set, and other nodes are follower nodes. $\mathcal{V}_T=\{3,4\}$ is the target node set.
Then
\[(-L)= \left[ \setlength{\arraycolsep}{2pt}{\begin{array}{*{20}{c}}
-1&1&0&0\\0&0&0&0\\
0&2&-2&0\\
0&0&2&-2\\
\end{array}} \right],
B= \left[\setlength{\arraycolsep}{2pt} {\begin{array}{*{20}{c}}
1\\
0\\
0\\0
\end{array}} \right]
, H= \left[ \setlength{\arraycolsep}{2pt}{\begin{array}{*{20}{c}}
0&0&1&0\\
0&0&0&1
\end{array}} \right],\]
$\det(\lambda I+L)=0$ for $\lambda_1=\lambda_2=-2,\lambda_3=-1,\lambda_4=0.$
\[\begin{aligned}
H[(-2)I-(-L),B]&=\left[ {\begin{array}{*{20}{c}}
0&-2&0&0&0\\
0&0&-2&0&0
\end{array}} \right];\\
H[(-1)I-(-L),B]&=\left[ {\begin{array}{*{20}{c}}
0&-2&1&0&0\\
0&0&-2&1&0
\end{array}} \right];\\
H[0I-(-L),B]&=\left[ {\begin{array}{*{20}{c}}
0&-2&2&0&0\\
0&0&-2&2&0
\end{array}} \right].
\end{aligned}\]
So
$$rank H[\lambda_iI-(-L),B]=2.\qquad i=1,2,3,4$$
In addition, there is no vector $q\in \mathbb{R}^p$ such that the nonzero left eigenvector $\vartheta^T$ of $-L$ satisfies $\vartheta^T=q^TH$ and  $\vartheta^TB=0$.

If the condition in Proposition \ref{pro2} is a sufficient one, then  system \eqref{2} is target controllable.
However, by Proposition  \ref{pro1}, if the target node is unreachable from the leader set,  system \eqref{2} is not target controllable. Obviously,  target nodes presented in Figure \ref{fig7} are unreachable from the leader set. So the system is not target controllable.

\begin{corollary}\label{cor3}
If there is a left eigenvector  $\vartheta^T$ of $-L$ such that $\vartheta^TB=0$, and the elements of $\vartheta^T$ corresponding to nontarget nodes are zero, then system \eqref{2}  is not target controllable.
\end{corollary}
\begin{proof}
Let the labels of target nodes be ${1},\ldots,{p}$, then $H=[e_{1},\cdots,e_{p}]^T$. From the item 2) of Proposition \ref{pro2},  if there is a vector $q\in \mathbb{R}^p$ such that the nonzero left eigenvector $\vartheta^T$ of $-L$ satisfies $\vartheta^T=q^TH$ and  $\vartheta^TB=0$, then the system is not target controllable. Next, we prove that the following two conditions are equivalent: 
\begin{enumerate}
\item the elements of $\vartheta^T$ corresponding to nontarget nodes are zero,
\item there is a vector $q\in \mathbb{R}^p$ such that $\vartheta^T$ satisfies $\vartheta^T=q^TH$.
\end{enumerate}
Assume that there is a vector $q\in \mathbb{R}^p$ such that $\vartheta^T$ satisfies $\vartheta^T=q^TH$,  one has
$$\begin{aligned}\vartheta^T&=q^TH\\&=[\bar{q}_1,\cdots,\bar{q}_p][e_{1},\cdots,e_{p}]^T\\&=\bar{q}_1e_{1}^T+\cdots+\bar{q}_pe_{p}^T\\&=[\bar{q}_1,\cdots,\bar{q}_p,0_{1\times (n-p)}].\end{aligned}$$
It can be seen that the elements of $\vartheta^T$ corresponding to the nontarget nodes are zero.
Assume that the elements of $\vartheta^T$ corresponding to the nontarget nodes are zero,  then $$\begin{aligned}\vartheta^T&=[\bar{q}_1,\cdots,\bar{q}_p,0_{1\times (n-p)}]\\&=\bar{q}_1e_{1}^T+\cdots+\bar{q}_pe_{p}^T\\&=[\bar{q}_1,\cdots,\bar{q}_p][e_{1},\cdots,e_{p}]^T\\&=[\bar{q}_1,\cdots,\bar{q}_p]H.\end{aligned}$$ Therefore, there is a vector $q^T=[\bar{q}_1,\cdots,\bar{q}_p]$ such that $\vartheta^T$ satisfies $\vartheta^T=q^TH$.
\end{proof}

Assume that system \eqref{2} is  incompletely controllable, it can be decomposed according to controllability \cite{ref26}.  There is a nonsingular transformation $\hat x=Px$, such that
\[\widehat{-L}=P(-L)P^{-1}=\left[\begin{array}{cc}\hat{L}_{c} & \hat{L}_{12} \\
0 & \hat{L}_{\bar{c}}
\end{array}\right],\]
\[\hat{B}=PB=\left[\begin{array}{c}
\hat{B}_{c} \\0
\end{array}\right],\hat{H}=HP^{-1}=[\hat{H}_{c},\hat{H}_{\bar{c}}],\]
where $\hat x_c\in\mathbb{R}^\kappa$ and $\hat x_{\bar c}\in\mathbb{R}^{n-\kappa}$ represent controllable state vector and uncontrollable state vector, respectively.
Let $P^{-1}=[\mathcal{P}_1,\mathcal{P}_2]$, where $\mathcal{P}_1\in {\mathbb{R}^{n\times \kappa}}$ and $\mathcal{P}_2\in {\mathbb{R}^{n\times (n-\kappa)}}$. ${p_{\gamma_1},p_{\gamma_2},\ldots,p_{\gamma_\kappa}}$ constitute the maximal linear independent group of rows in $\mathcal{P}_1$, where $p_{\gamma_1},p_{\gamma_2},\ldots,p_{\gamma_\kappa}$ are the $\gamma_1$-th, $\gamma_2$-th, $\ldots$, $\gamma_\kappa$-th  rows of the matrix $\mathcal{P}_1$, respectively. At the same time, $\gamma_1,\gamma_2,\ldots,\gamma_\kappa$ are the labels of the maximal linear independent rows $p_{\gamma_1},p_{\gamma_2},\ldots,p_{\gamma_\kappa}$.

\begin{lemma}\label{lem2}
System \eqref{2} is target controllable if and only if $rank\hat{H}_{c}=p$.
\end{lemma}
\begin{proof}
(Sufficiency). Suppose $rank\hat{H}_{c}=p$. The target controllability matrix is
$$\begin{aligned}W&=H[B,(-L)B,(-L)^2B,\ldots,(-L)^{n-1}B]\\&=HP^{-1}[PB,P(-L)P^{-1}PB,\ldots,(P(-L)P^{-1})^{n-1}PB]\\&=[
\hat{H}_{c},\hat{H}_{\bar{c}}]\left[ {\begin{array}{*{20}{c}}
\hat{B}_{c}&\hat{L}_{c}\hat{B}_{c}&\hat{L}^2_{c}\hat{B}_{c}&\ldots&\hat{L}^{n-1}_{c}\hat{B}_{c}\\
0&0&0&\ldots&0
\end{array}} \right]\\&=\left[ {\begin{array}{*{20}{c}}
\hat{H}_{c}\hat{B}_{c}&\hat{H}_{c}L_{c}B_{c}&\hat{H}_{c}L^2_{c}B_{c}&\ldots&\hat{H}_{c}L^{n-1}_{c}B_{c}
\end{array}} \right]\\&=\hat{H}_{c}\left[ {\begin{array}{*{20}{c}}
\hat{B}_{c}&\hat{L}_{c}\hat{B}_{c}&\hat{L}^2_{c}\hat{B}_{c}&\ldots&\hat{L}^{n-1}_{c}\hat{B}_{c}
\end{array}} \right].
\end{aligned}$$
 It is known that $\hat{L}_{c}$ and $\hat{B}_{c}$ are the system matrix and input matrix corresponding to the controllable substates, respectively. So $\left[ {\begin{array}{*{20}{c}}
\hat{B}_{c}&\hat{L}_{c}\hat{B}_{c}&\hat{L}^2_{c}\hat{B}_{c}&\ldots&\hat{L}^{n-1}_{c}\hat{B}_{c}
\end{array}} \right]$ is full row rank. Then, $rank \hat{H}_{c}\left[ {\begin{array}{*{20}{c}}
\hat{B}_{c}&\hat{L}_{c}\hat{B}_{c}&\hat{L}^2_{c}\hat{B}_{c}&\ldots&\hat{L}^{n-1}_{c}\hat{B}_{c}
\end{array}} \right]=rank \hat{H}_{c}=p$. Namely, $rankW=p$. Thus, system \eqref{2} is target controllable.

(Necessity).
Assume that  system \eqref{2} is target controllable. Then, by Lemma \ref{lem1}, $rankW=p$. Namely, $$rank \left(\hat{H}_{c}\left[ {\begin{array}{*{20}{c}}
\hat{B}_{c}&\hat{L}_{c}\hat{B}_{c}&\hat{L}^2_{c}\hat{B}_{c}&\ldots&\hat{L}^{n-1}_{c}\hat{B}_{c}
\end{array}} \right]\right)=p.$$ Since $\left[ {\begin{array}{*{20}{c}}
\hat{B}_{c}&\hat{L}_{c}\hat{B}_{c}&\hat{L}^2_{c}\hat{B}_{c}&\ldots&\hat{L}^{n-1}_{c}\hat{B}_{c}
\end{array}} \right]$ is full row rank, $rank\hat{H}_{c}=p$.
\end{proof}

\begin{theorem}\label{theo2}
System \eqref{2} is target controllable if and only if target nodes are chosen from  $\gamma_1,\gamma_2,\ldots,\gamma_\kappa$, where $\gamma_1,\gamma_2,\ldots,\gamma_\kappa$ are the labels of the maximal linear independent rows $p_{\gamma_1},p_{\gamma_2},\ldots,p_{\gamma_\kappa}$.
\end{theorem}

\begin{proof}
(Sufficiency). Assume that the labels of target nodes are $\delta_1,\delta_2,\ldots,\delta_p$, where $\{\delta_1,\delta_2,\ldots,\delta_p\}\subseteq\{\gamma_1,\gamma_2,\ldots,\gamma_\kappa\}$. Then, $\{{p_{\delta_1},p_{\delta_2},\ldots,p_{\delta_p}}\}\subseteq\{{p_{\gamma_1},p_{\gamma_2},\ldots,p_{\gamma_\kappa}}\}$.  Since ${p_{\gamma_1},p_{\gamma_2},\ldots,p_{\gamma_\kappa}}$ constitute the maximal linear independent group of rows in $\mathcal{P}_1$, ${p_{\delta_1},p_{\delta_2},\ldots,p_{\delta_p}}$ are linearly independent.
$$\begin{aligned}\hat{H}_{c}&=H\mathcal{P}_1\\&=\left[ {\begin{array}{*{20}{c}}
e_{\delta_1}&e_{\delta_2}& \ldots&e_{\delta_p}
\end{array}} \right]^T\left[ {\begin{array}{*{20}{c}}
\bar{p}_1&\bar{p}_2&\ldots&\bar{p}_\kappa
\end{array}} \right]\\&=\left[ {\begin{array}{*{20}{c}}
e_{\delta_1}^T\bar{p}_1&e_{\delta_1}^T\bar{p}_2&\ldots&e_{\delta_1}^T\bar{p}_\kappa\\
e_{\delta_2}^T\bar{p}_1&e_{\delta_2}^T\bar{p}_2&\ldots&e_{\delta_2}^T\bar{p}_\kappa\\
\vdots&\vdots&\ddots&\vdots\\
e_{\delta_p}^T\bar{p}_1&e_{\delta_p}^T\bar{p}_2&\ldots&e_{\delta_p}^T\bar{p}_\kappa
\end{array}} \right]\\&=\left[ {\begin{array}{*{20}{c}}
p_{\delta_1}^T&p_{\delta_2}^T&\ldots& \ p_{\delta_p}^T
\end{array}} \right]^T,
\end{aligned}$$
where $\bar{p}_1,\bar{p}_2,\ldots,\bar{p}_\kappa$ are the columns of $\mathcal{P}_1$. So  $rank \hat{H}_{c}=p$. By Lemma \ref{lem2}, the system is target controllable. 

(Necessity). ${p_{\gamma_1},p_{\gamma_2},\ldots,p_{\gamma_\kappa}}$ constitute the maximal linear independent group of rows in $\mathcal{P}_1$. Therefore, if target nodes are not all selected from $\gamma_1,\gamma_2,\ldots,\gamma_\kappa$, then $rank \hat{H}_{c}<p$. By Lemma \ref{lem2}, system \eqref{2} is not target controllable.
\end{proof}

\begin{remark}
 Lemma \ref{lem2} reveals the relationship between the target controllability  and  controllability decomposition of the system. The controllable substate is $\hat{x}_{c}\in\mathbb{R}^\kappa$. If the number of target nodes satisfies $p>\kappa$, the system is not target controllable. By controllability decomposition, Theorem 2 shows a method to select target nodes to ensure target controllability. The relevant Example 5 is in Section \ref{sec7}.
\end{remark}

\section{target controllability of a high-order multiagent system}\label{sec4}

In this section, a high-order multiagent system is considered. The state of each agent is given by the following $m$ order differential equation, where $i={1,2,\cdots,n}$.

\begin{eqnarray}\label{4}
\begin{cases}
\dot{x}_i(t) =x^{(1)}_i(t),   \\
\dot{x}^{(1)}_i(t)=x^{(2)}_i(t),  \\
   \qquad \,\,\vdots          \\
\dot{x}^{(m-1)}_i(t)=\sum\limits_{j \in {\mathcal{N}_i}}^{}{ a_{ij}}{[{x_j(t)}-{x_i(t)}]}+u_i(t),& i\in  \mathcal{V}_L \\
\dot{x}^{(m-1)}_i(t)=\sum\limits_{j \in {\mathcal{N}_i}}^{}{ a_{ij}}{[{x_j(t)}-{x_i(t)}]},  & i\in \mathcal{V} _F \\
y_i(t)=x_i(t), & i\in  \mathcal{V}_T.
\end{cases}
\end{eqnarray}
Dynamics \eqref{4} can be rewritten as

\begin{equation}
\begin{split}
\label{5}
\dot{X}(t)&=\mathcal{A}X(t)+\mathcal{B}U(t),\\
Y(t)&=\mathcal{H}X(t),
\end{split}
\end{equation}
where $X(t)=[x_1(t),\cdots,x_n(t),x^{(1)}_1(t),\cdots,x^{(1)}_n(t),\cdots,$\\$x^{(m-1)}_1(t),\cdots,x^{(m-1)}_n(t)]^T$, and
$U(t)=[\tilde{u}_{1}(t),\cdots,\tilde{u}_{l}(t)]^T$.  $\tilde{u}_k$ is the $k$-th element of control input $U$, where $k= {1,2,\cdots,l}$.

$\mathcal{A}=\left[ {\begin{array}{*{20}{c}}0_n&I_n&0_n&\cdots&0_n\\ 0_n&0_n&I_n&\cdots&0_n\\ \vdots&\vdots&\vdots& \ddots&\vdots\\ 0_n&0_n&0_n&\cdots&I_n\\  -L&0_n&0_n&\cdots&0_n \end{array}} \right],
\mathcal{B}=\left[ {\begin{array}{*{20}{c}}0_{n\times l}\\ \vdots \\0_{n\times l}\\ B\end{array}} \right],$$$
\mathcal{H}=\left[ {\begin{array}{*{20}{c}}H &0_{p\times n}&\cdots&0_{p\times n} \end{array}} \right],$$
where $\mathcal{A}\in \mathbb{R}^{nm\times nm}$, $\mathcal{B}\in \mathbb{R}^{nm\times l}$, $\mathcal{H}\in \mathbb{R}^{p\times nm}$.

\begin{theorem}\label{theo3}
Once the topology, leaders, and target nodes are fixed, the high-order multiagent system  \eqref{5} is target controllable if and only if so is the  first-order multiagent system \eqref{2}.
\end{theorem}

\begin{proof}
By calculation, one has\\
\begin{strip}
 \hrulefill
\[
\begin{aligned}
I&=\setlength{\arraycolsep}{1.5pt}{\tiny {\addtocounter{MaxMatrixCols}{32}
\begin{bmatrix}
I_n&0_n&\cdots&0_n&0_n\\ 0_n&I_n&\cdots&0_n&0_n\\ \vdots&\vdots&\ddots& \vdots& \vdots\\ 0_n&0_n&\cdots&I_n&0_n \\ 0_n&0_n&\cdots&0_n&I_n
\end{bmatrix}}},&\mathcal{A}^m&=\setlength{\arraycolsep}{1.5pt}{\tiny {\addtocounter{MaxMatrixCols}{32}
\begin{bmatrix}
-L&0_n&\cdots&0_n&0_n\\ 0_n&-L&\cdots&0_n&0_n\\ \vdots&\vdots&\ddots& \vdots&\vdots\\ 0_n&0_n&\cdots&-L&0_n\\ 0_n&0_n&\cdots&0_n&-L
\end{bmatrix}}},\cdots,&\mathcal{A}^{(n-1)m}&=\setlength{\arraycolsep}{0.5pt}{\tiny {\addtocounter{MaxMatrixCols}{32}
\begin{bmatrix}
(-L)^{n-1}&0_n&\cdots&0_n&0_n\\ 0_n&(-L)^{n-1}&\cdots&0_n&0_n\\ \vdots&\vdots&\ddots&\vdots&\vdots\\ 0_n&0_n&\cdots&(-L)^{n-1}&0_n\\ 0_n& 0_n&\cdots&0_n&(-L)^{n-1}
\end{bmatrix}}},\\
\mathcal{A}&=\setlength{\arraycolsep}{1.5pt}{\tiny {\addtocounter{MaxMatrixCols}{32}
\begin{bmatrix}
0_n&I_n&0_n&\cdots&0_n\\ 0_n&0_n&I_n&\cdots&0_n\\ \vdots&\vdots&\vdots& \ddots&\vdots\\ 0_n&0_n&0_n&\cdots&I_n\\ -L&0_n&0_n&\cdots&0_n
\end{bmatrix}}}, &\mathcal{A}^{m+1}&=\setlength{\arraycolsep}{1.5pt}{\tiny {\addtocounter{MaxMatrixCols}{32}
\begin{bmatrix}
0_n&-L&0_n&\cdots&0_n\\ 0_n&0_n&-L&\cdots&0_n\\ \vdots&\vdots&\vdots& \ddots&\vdots\\ 0_n&0_n&0_n&\cdots&-L\\ (-L)^2&0_n&0_n&\cdots&0_n
\end{bmatrix}}},\cdots,&\mathcal{A}^{(n-1)m+1}&=\setlength{\arraycolsep}{1.5pt}{\tiny  {\addtocounter{MaxMatrixCols}{32}
\begin{bmatrix}
0_n&(-L)^{n-1}&0_n&\cdots&0_n\\ 0_n&0_n&(-L)^{n-1}&\cdots&0_n\\ \vdots&\vdots&\vdots& \ddots&\vdots\\ 0_n&0_n&0_n&\cdots&(-L)^{n-1}\\ (-L)^{n}&0_n&0_n&\cdots&0_n
\end{bmatrix}}},\\
&\vdots&&\vdots\qquad\qquad\qquad\qquad\qquad\quad\vdots&&\vdots\\
\mathcal{A}^{m-2}&=\setlength{\arraycolsep}{1.5pt}{\tiny  {\addtocounter{MaxMatrixCols}{32}
\begin{bmatrix}
0_n&\cdots&0_n&I_n&0_n\\0_n&\cdots&0_n&0_n&I_n\\ -L&\cdots&0_n&0_n&0_n\\ \vdots& \ddots&\vdots&\vdots&\vdots\\  0_n&\cdots&-L&0_n&0_n
\end{bmatrix}}},&\mathcal{A}^{2m-2}&=\setlength{\arraycolsep}{1.5pt}{\tiny  {\addtocounter{MaxMatrixCols}{32}
\begin{bmatrix}
0_n&\cdots&0_n&-L&0_n\\0_n&\cdots&0_n&0_n&-L\\ (-L)^2&\cdots&0_n&0_n&0_n\\ \vdots&\ddots&\vdots&\vdots&\vdots\\  0_n&\cdots&(-L)^2&0_n&0_n
\end{bmatrix}}},\cdots,&\mathcal{A}^{nm-2}&=\setlength{\arraycolsep}{1.5pt}{\tiny  {\addtocounter{MaxMatrixCols}{32}
\begin{bmatrix}
0_n&\cdots&0_n&(-L)^{n-1}&0_n\\0_n&\cdots&0_n&0_n&(-L)^{n-1}\\ (-L)^n&\cdots&0_n&0_n&0_n\\ \vdots& \ddots&\vdots&\vdots&\vdots\\  0_n&\cdots&(-L)^n&0_n&0_n
\end{bmatrix}}},\\
\mathcal{A}^{m-1}&=\setlength{\arraycolsep}{1.5pt}{\tiny  {\addtocounter{MaxMatrixCols}{32}
\begin{bmatrix}
0_n&0_n&\cdots&0_n&I_n\\-L&0_n&\cdots&0_n&0_n\\ 0_n&-L&\cdots&0_n&0_n\\ \vdots&\vdots& \ddots&\vdots&\vdots\\  0_n&0_n&\cdots&-L&0_n
\end{bmatrix}}},&\mathcal{A}^{2m-1}&=\setlength{\arraycolsep}{1.5pt}{\tiny  {\addtocounter{MaxMatrixCols}{32}
\begin{bmatrix}
0_n&0_n&\cdots&0_n&-L\\(-L)^2&0_n&\cdots&0_n&0_n\\ 0_n&(-L)^2&\cdots&0_n&0_n\\ \vdots&\vdots& \ddots&\vdots&\vdots\\  0_n&0_n&\cdots&(-L)^2&0_n
\end{bmatrix}}},\cdots,&\mathcal{A}^{nm-1}&=\setlength{\arraycolsep}{1.5pt}{\tiny {\addtocounter{MaxMatrixCols}{32}
\begin{bmatrix}
0_n&0_n&\cdots&0_n&(-L)^{n-1}\\(-L)^n&0_n&\cdots&0_n&0_n\\ 0_n&(-L)^n&\cdots&0_n&0_n\\ \vdots&\vdots& \ddots&\vdots&\vdots\\  0_n&0_n&\cdots&(-L)^n&0_n
\end{bmatrix}}}.
\end{aligned}
\]
\hrulefill
\end{strip}
\noindent
Then the target controllability matrix is
$$
\begin{aligned}
  W&=[\mathcal{H}\mathcal{B},\mathcal{H}\mathcal{A}\mathcal{B},\ldots,\mathcal{H}\mathcal{A}^{nm-1} \mathcal{B}]\\
  &=\mathcal{H} \left[\setlength{\arraycolsep}{0.01pt} {\begin{array}{*{20}{c}}0_{n\times l}&\cdots&B&\cdots&0_{n\times l}&\cdots&(-L)^{n-1}B\\
 \vdots&\ddots&\vdots&\cdots&\vdots&\ddots&\vdots\\
B&\cdots&0_{n\times l}&\cdots&(-L)^{n-1}B&\cdots&0_{n\times l}\end{array}} \right]\\
  &=[0_{p\times l},\cdots,HB,\cdots,0_{p\times l},\cdots,H(-L)^{n-1}B].
\end{aligned}
$$
It follows that $rankW=rank[HB,\cdots,H(-L)^{n-1}B]$. By Lemma  \ref{lem1}, once the topology, leaders, and target nodes are fixed, the target controllability of the high-order multiagent system is shown to be the same to the first-order one. This completes the proof. 
\end{proof}

\begin{remark}
 By Theorem \ref{theo3}, the results about the first-order multiagent system in Section \ref{sec3}
 also apply to the high-order system.
\end{remark}

\section{target controllability of a general linear system}\label{sec5}
There are many studies based on first-order multiagent systems. But in practical application, we need to deal with more complex models. Consider the following general linear system,

\begin{align}
\begin{cases}\label{6}
\dot{\tilde{x}}_i(t)=\tilde{A}\tilde{x}_i(t)+Mz_i(t)+N\bar{u}_i(t) ,& i\in  \mathcal{V}_L\\
\dot{\tilde{x}}_i(t)=\tilde{A}\tilde{x}_i(t)+Mz_i(t), & i\in  \mathcal{V}_F \\
\tilde{y}_i(t)=C\tilde{x}_i(t) ,& i\in \mathcal{V}_T
\end{cases}
\end{align}
\[\dot{z}_i(t)=K\sum\limits_{j \in {\mathcal{N}_i}}^{}{ a_{ij}}{[{\tilde{x}_j(t)} - {\tilde{x}_i(t)}]},\]
where $\tilde{x}_i$ represents the state of agent  $i$, $\bar{u}_i$ is the external input, and $K$ is a feedback gain matrix. Consider the special case that $C$ is an identity matrix and the output vector is the state vector of target nodes.   The compact form of \eqref{6} is 

\begin{equation}
\begin{split}
\label{7}
\dot{\tilde{x}}(t)&=\tilde{L}\tilde{x}(t)+\tilde{B}\bar{u}(t),\\
\tilde{y}(t)&=\tilde{H}\tilde{x}(t),
\end{split}
\end{equation}
where $\tilde{L}=I_n\otimes \tilde{A}-L\otimes MK$;  $\tilde{H}=H\otimes I_\sigma$,  $H = [e_{\bar{v}_1},\cdots,e_{\bar{v}_p}] ^T\in \mathbb{R}^{p\times n}$; and $\tilde{B}=B\otimes N$,  $B = [e_{v_1},\cdots,e_{v_l}]$.

\begin{definition}
 A directed graph $\mathcal{G}$ is called leader-target follower connected, if there is an independent strongly connected component $\mathcal{G}^{c_{\bar{i}}}_f$  in $\mathcal{G}_f$ that contains only target nodes  and there is an edge from the leader to a node in $\mathcal{G}^{c_{\bar{i}}}_f$, where $\mathcal{G}_f$ is an induced subgraph that contains all follower nodes.
\end{definition}

\begin{figure}[t]
\centerline{\includegraphics[width=150pt,height=90pt]{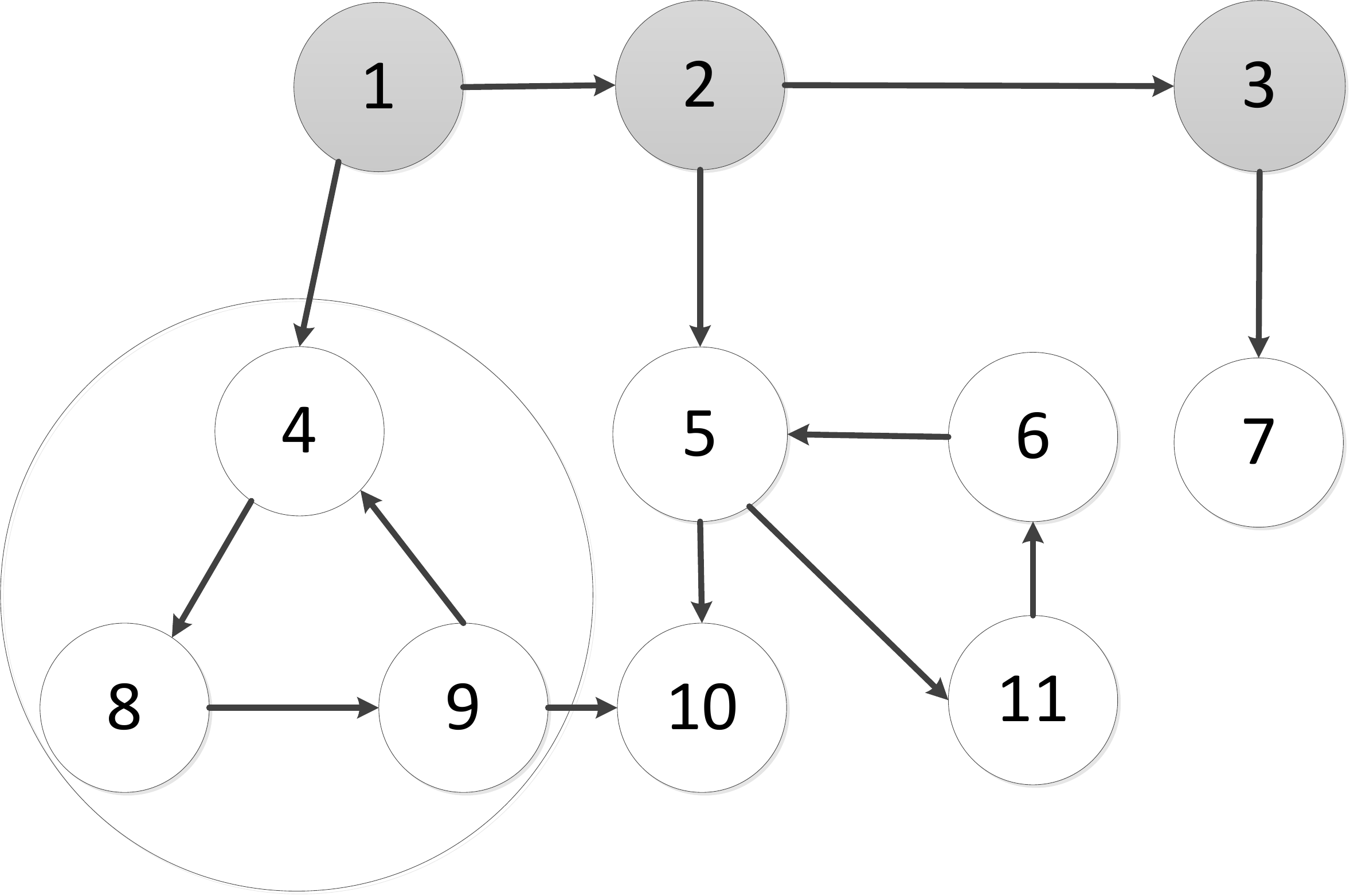}}
\caption{A leader-target follower connected graph.\label{fig9}}
\end{figure}

In Figure \ref{fig9}, $\mathcal{V}_L=\{1,2,3\}$, $\mathcal{V}_T=\{4,8,9,10\}$. There is an independent strongly connected component $\mathcal{G}^{c_{\bar{i}}}_f$ that contains only target nodes, and $\mathcal{G}^{c_{\bar{i}}}_f=\{4,8,9\}$. There is an edge from the leader to one node in $\{4,8,9\}$, and  Figure \ref{fig9} is called leader-target follower connected.

\begin{lemma}\label{lem5}
If there is a left eigenvector  $\vartheta^T$ of $\tilde{L}$ such that $\vartheta^T\tilde{B}=0$, then system \eqref{7} is not target controllable, where the elements of $\vartheta^T$ corresponding to the nontarget nodes are zero.
\end{lemma}
\begin{proof}
Lemma \ref{lem5} can be proved similarly as Corollary \ref{cor3}.
\end{proof}

\begin{proposition}\label{pro5}
If there is an independent strongly connected component that contains only target nodes and the system \eqref{7} is target controllable, then graph $\mathcal{G}$ is leader-target follower connected.
\end{proposition}

\begin{proof}
Let $L_i$ represent the submatrix associated with $\mathcal{G}^{c_i}_f$ in $L$, where $\mathcal{G}^{c_i}_f$ is  an independent strongly connected component in $\mathcal{G}_f$,  $i=1,\cdots,\bar{\gamma}$. $L_{\bar{\gamma}+1}$ represents the submatrix associated with nonindependent strongly connected components in $L$, and $L_l$ represents the submatrix associated with leaders in $L$. Assume that the graph $\mathcal{G}$ is not leader-target follower connected. In one case, there is no independent strongly connected component $\mathcal{G}^{c_{\bar{i}}}_f$ that contains only target nodes. In the other case, there is an independent strongly connected component $\mathcal{G}^{c_{\bar{i}}}_f$ that contains only the target node, but there is no edge from any leader to the nodes in the $\mathcal{G}^{c_{\bar{i}}}_f$. Let $L_1$ represent the submatrix associated with $\mathcal{G}^{c_{\bar{i}}}_f$ in $L$, where $L_1\in \mathbb{R}^{a\times a}$.
Then $L$ and $B$ can be expressed in the following form.
$$
L=\left[\setlength{\arraycolsep}{0.01pt}\begin{array}{c|ccccc}
L_1&0&\cdots&0&0&0\\ \hline
 0&L_2&\cdots&0&0&\ast\\
\vdots&\vdots&\ddots& \vdots&\vdots&\vdots\\
0&0&\cdots&L_{\bar{\gamma}}&0&\ast\\
\ast&\ast&\cdots&\ast&L_{\bar{\gamma}+1}&\ast\\
\ast&\ast&\cdots&\ast&\ast&L_l
\end{array}\right]=\left[ \setlength{\arraycolsep}{0.01pt}{\begin{array}{*{20}{c}}L_1&0\\ \ast&L_\chi \end{array}} \right], 
B=\left[\setlength{\arraycolsep}{0.01pt}\begin{array}{c}
0\\ \hline
 0\\ \vdots\\ 0\\ 0\\B_1
\end{array}\right]=\left[ \setlength{\arraycolsep}{0.01pt}{\begin{array}{*{20}{c}}0\\B_\chi\end{array}} \right].
$$
Let ${\tilde{\vartheta}}^T$ be a left eigenvector of the matrix $I_a \otimes \tilde{A}-L_1\otimes MK$, with an eigenvlaue being $\lambda_1$. There is ${\vartheta} ^T=[{\tilde{\vartheta}}^T,0^T]$ such that
$$
\begin{aligned}
{{\vartheta} ^T}\tilde{L} &=[{\tilde{\vartheta}}^T,0^T] (I_n \otimes \tilde{A}-L \otimes MK) \\
&=[{\tilde{\vartheta}}^T,0^T]\left \{ \left[ {\begin{array}{*{20}{c}}I_a&0_{a \times (n-a)}\\0_{(n-a )\times a }&I_{n-a  } \end{array}} \right]\otimes \notag\right.
\\&
\phantom{=\;\;}
\left.
\tilde{A}-\left[ {\begin{array}{*{20}{c}}L_1&0_{a \times (n-a)}\\ \ast&L_\chi \end{array}} \right] \otimes MK \right\}\\
&=[{\tilde{\vartheta}}^T,0^T]\left[\setlength{\arraycolsep}{0.001pt} \small{{\begin{array}{*{20}{c}}I_a \otimes \tilde{A}-L_1\otimes MK&0\\-(\ast\otimes MK)&I_{n-a} \otimes \tilde{A}-L_{\chi}\otimes MK \end{array}}} \right] \\
&=\left[ {\begin{array}{*{20}{c}}{\tilde{\vartheta}}^T(I_a \otimes \tilde{A}-L_1\otimes MK)&0\end{array}} \right]\\
&=\lambda_1[\tilde{\vartheta}^T,0]\\
&=\lambda_1\vartheta^T.
\end{aligned}
$$
$$
\begin{aligned}
{{\vartheta} ^T}\tilde{B} &=[{\tilde{\vartheta}}^T,0^T](B\otimes N)=[{\tilde{\vartheta}}^T,0]\left[ {\begin{array}{*{20}{c}}0\\B_\chi\otimes N\end{array}} \right]=0.
\end{aligned}
$$
Namely, there is a left eigenvector  $\vartheta^T$ of $\tilde{L}$ such that $\vartheta^T\tilde{B}=0$. 
Since ${\vartheta} ^T=[{\tilde{\vartheta}}^T,0^T]$, the elements of $\vartheta^T$ corresponding to the nontarget nodes are zero. By Lemma \ref{lem5}, system \eqref{7} is not target controllable, which contradicts the target controllability of system \eqref{7}. 

In summary, if there is an independent strongly connected component that contains only target nodes and the system \eqref{7} is target controllable, the graph $\mathcal{G}$ is leader-target follower connected.
\end{proof}

\begin{corollary}\label{cor4}
If there is an independent strongly connected component that contains only target nodes and the system \eqref{2} is target controllable, the graph $\mathcal{G}$ is leader-target follower connected.
\end{corollary}

\begin{proof}
Through comparative analysis, when $\sigma =1$, $\tilde{A}=0,M=1,N=1,C=1,K=1$, the general linear system \eqref{7} can be transformed into the first-order multiagent system \eqref{2}. Therefore, Corollary \ref{cor4} follows from Proposition \ref{pro5}.
\end{proof}

\section{illustrative examples}\label{sec7}
\begin{example}
\begin{figure}[t]
\centerline{\includegraphics[width=150pt,height=70pt]{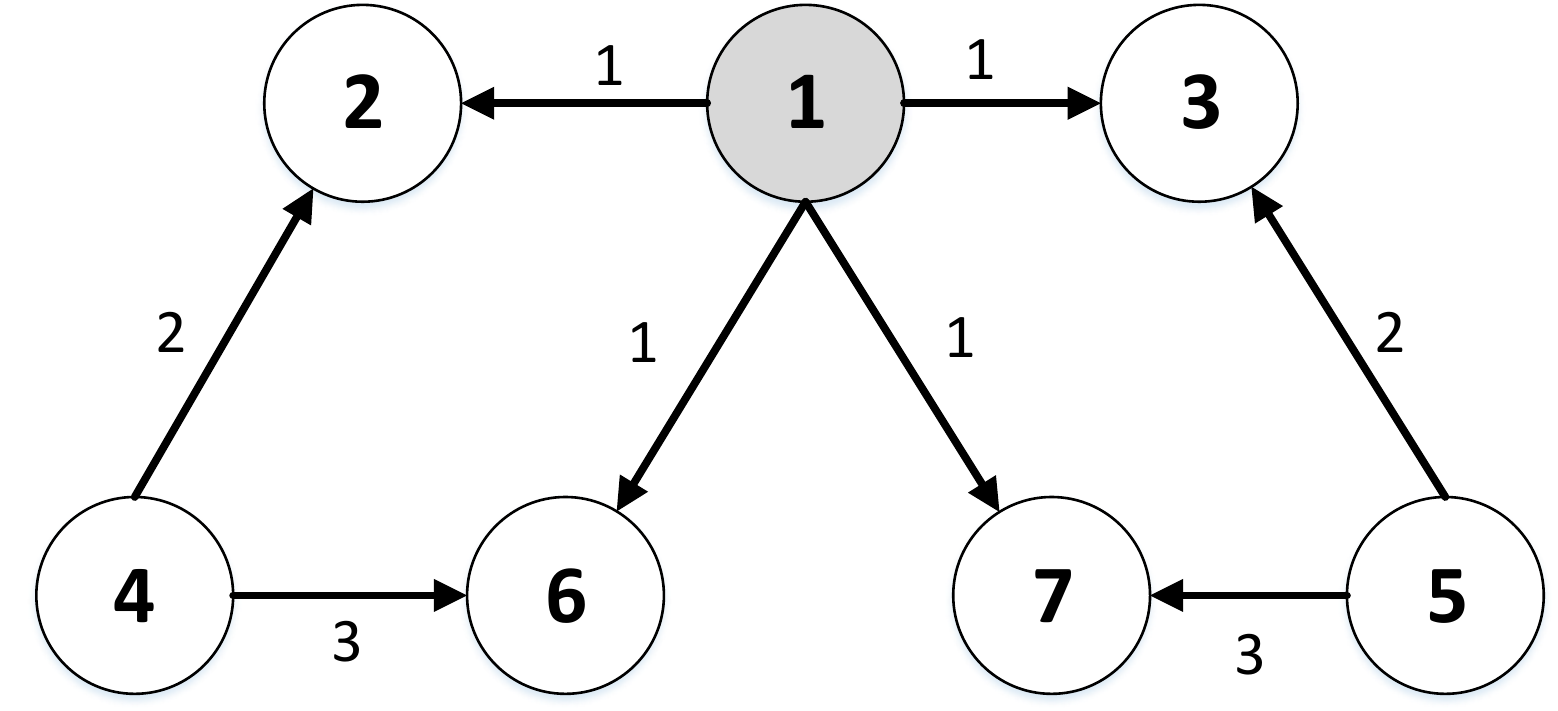}}
\caption{A directed weighted topology graph with seven nodes.\label{fig2}}
\end{figure}
In Figure \ref{fig2}, $\mathcal{V}_L=\{1\}$ is the leader node set, and other nodes are follower nodes, with $\mathcal{V}_T=\{2,6\}$ being the target node set. All target nodes in Figure \ref{fig2} are reachable from the leader set. 
From Figure \ref{fig2},
$$(-L)= \left[ {\begin{array}{*{20}{c}}
0&0&0&0&0&0&0\\
1&-3&0&2&0&0&0\\
1&0&-3&0&2&0&0\\
0&0&0&0&0&0&0\\
0&0&0&0&0&0&0\\
1&0&0&3&0&-4&0\\
1&0&0&0&3&0&-4
\end{array}} \right],
B= \left[ {\begin{array}{*{20}{c}}
1\\
0\\
0\\
0\\
0\\
0\\
0
\end{array}} \right]
,$$$$ H= \left[ {\begin{array}{*{20}{c}}
0&1&0&0&0&0&0\\
0&0&0&0&0&1&0
\end{array}} \right],$$\vspace*{1mm}
$$\begin{aligned}
rankW&=rank[HB,H(-L)B,\cdots,H(-L)^{6}B] \\&=rank\left[ \small{{\begin{array}{*{20}{c}}
0&1&-3&9&-27&81&-243\\
0&1&-4&16&-64&256&-1024
\end{array}}} \right] \\&=2.
\end{aligned}$$
Therefore, in this example, system \eqref{2} is target controllable. It can be obtained from Figure \ref{fig2} that all target nodes are reachable from the leader set, which is consistent with  Proposition \ref{pro1}.
The simulation of the first-order multiagent system in Figure \ref{fig2} is shown in Figure \ref{fig3}, where the squares and stars represent the initial and terminal states, respectively.
\begin{figure}[t]
\centerline{\includegraphics[width=150pt,height=120pt]{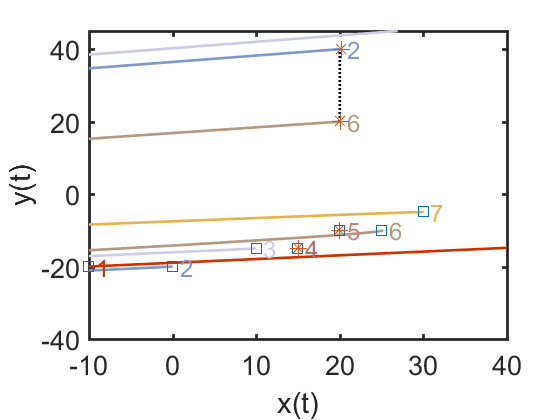}}
\caption{A target controllable graph of the first-order system.\label{fig3}}
\end{figure}
\end{example}

\begin{example}
 Let
$$\begin{aligned}
r_\delta=
\begin{cases}
\tilde{r}_{\delta}, &\delta=0\cr
r_{\delta-1}+\tilde{r}_{\delta}, &\delta=1,2,\cr \end{cases}
\end{aligned}
$$
where $\tilde{r}_{\delta}$ is the number of $\delta$-reachable nodes. 
In Figure \ref{fig6}, ${\mathcal{V}_L}=\{1,2,3\}$ is the leader node set, ${\mathcal{V}^{r}_L}=\{4 ,5,6,7,8\}$ is reachable node set, and  ${\mathcal{V}_D}=\{9\}$ is the unreachable node set. Nodes $4$ and $5$ are $0$-reachable nodes, namely $\tilde{r}_0=2$; nodes $6$ and $7$ are $1$-reachable nodes, namely $\tilde{r}_1=2$; and the node $8$ is a $2-$reachable node, namely $\tilde{r}_2=1$. 
$$L=\left[\setlength{\arraycolsep}{2pt}\begin{array}{ccc|ccccc|c}
1&0&0&0&0&0&0&0&-1\\
-1&2&-1&0&0&0&0&0&0\\
0&0&0&0&0&0&0&0&0\\ \hline
-1&-2&0&3&0&0&0&0&0\\
-1&-1&0&-1&3&0&0&0&0\\
0&0&0&-2&0&3&0&0&-1\\
0&0&0&0&-1&-1&2&0&0\\
0&0&0&0&0&0&-1&1&0\\ \hline
0&0&0&0&0&0&0&0&0\\
\end{array}\right],B=\left[\setlength{\arraycolsep}{2pt}\begin{array}{ccccccccc}
1&0&0\\
0&1&0\\
0&0&1\\
0&0&0\\
0&0&0\\
0&0&0\\
0&0&0\\
0&0&0\\
0&0&0\\
\end{array}\right].$$
$$L_{22}=\left[\begin{array}{cc|cc|c}
-3&0&0&0&0\\
1&-3&0&0&0\\ \hline
2&0&-3&0&0\\
0&1&1&-2&0\\ \hline
0&0&0&1&-1\\
\end{array}\right],L_{21}=\left[\begin{array}{ccc}
1&2&0\\
1&1&0\\ \hline
0&0&0\\
0&0&0\\
0&0&0\\
\end{array}\right].$$
where $\iota_{\tilde{\epsilon}_04},\iota_{\tilde{\epsilon}_05},\iota_{\tilde{\epsilon}_06}$ are all nonnegative and  $\iota_{\tilde{\epsilon}_04},\iota_{\tilde{\epsilon}_05}$ are  positive, $\tilde{\epsilon}_0=4,5$; $\beta^{(1)}_{\epsilon_11},\beta^{(1)}_{\epsilon_12}$ are all nonnegative and $\beta^{(1)}_{31},\beta^{(1)}_{42}$ are positive, $\epsilon_1=3,4$; and $\beta^{(1)}_{53},\beta^{(1)}_{54}$ are all nonnegative and $\beta^{(1)}_{54}$ is positive.
$$\begin{aligned}
Q_{et}=\left[ {\begin{array}{*{20}{c}}
 I_3&0&0&0&\ldots&0\\
 0&L_{21}&L_{22}L_{21}&L_{22}^2L_{21}&\ldots&L_{22}^{7}L_{21}\\
 0&0&0&0&\ldots&0\\
\end{array}} \right]
\end{aligned}.$$

There is a block $I_3$, which consists of elements at the intersection of the first $3$ rows and the first $3$ columns of the matrix $Q_{et}$. So the rows $1$ to $3$ of the matrix $Q_{et}$ are nonzero rows.

\begin{enumerate}
\item $\iota_{44},\iota_{45}$ are positive, hence rows $4$ to $5$ of matrix $Q_{et}$ are nonzero rows. 
\item $\beta^{(1)}_{31}$ is positive, then   $\beta^{(1)}_{31}\iota_{\tilde{\epsilon}_04}$, $\beta^{(1)}_{31}\iota_{\tilde{\epsilon}_05}$  are positive.
By calculating $L_{22}L_{21}$,  $\iota_{67}=\beta^{(1)}_{31}\iota_{44}+\beta^{(1)}_{32}\iota_{54}$, $\iota_{68}=\beta^{(1)}_{31}\iota_{45}+\beta^{(1)}_{32}\iota_{55}$. $\beta^{(1)}_{31}$, $\beta^{(1)}_{32}$, $\iota_{\tilde{\epsilon}_04}$, $\iota_{\tilde{\epsilon}_05}$, $\iota_{\tilde{\epsilon}_06}$ are nonnegative. So  $\beta^{(1)}_{31}\iota_{44}$, $\beta^{(1)}_{32}\iota_{54}$, $\beta^{(1)}_{31}\iota_{45}$, $\beta^{(1)}_{32}\iota_{55}$ are nonnegative. 
 There are $\beta^{(1)}_{31}\iota_{44}$, $\beta^{(1)}_{31}\iota_{45}$ in $\beta^{(1)}_{31}\iota_{44}$, $\beta^{(1)}_{32}\iota_{54}$, $\beta^{(1)}_{31}\iota_{45}$, $\beta^{(1)}_{32}\iota_{55}$, and $\beta^{(1)}_{31}\iota_{44}$, $\beta^{(1)}_{31}\iota_{45}$ are positive. Since   $\beta^{(1)}_{31}\iota_{44}$, $\beta^{(1)}_{32}\iota_{54}$, $\beta^{(1)}_{31}\iota_{45}$, $\beta^{(1)}_{32}\iota_{55}$ are  nonnegative, then  $\iota_{67}$ and $\iota_{68}$ are positive. Therefore, the  $6$-th row of matrix $Q_{et}$ is a nonzero row.  

$\beta^{(1)}_{42}$ is positive, then 
 each value in $\beta^{(1)}_{42}\iota_{\tilde{\epsilon}_04}$, $\beta^{(1)}_{42}\iota_{\tilde{\epsilon}_05}$ is positive. By calculating $L_{22}L_{21}$,  $\iota_{77}=\beta^{(1)}_{41}\iota_{44}+\beta^{(1)}_{42}\iota_{54}$, $\iota_{78}=\beta^{(1)}_{41}\iota_{45}+\beta^{(1)}_{42}\iota_{55}$, $\beta^{(1)}_{41}$, $\beta^{(1)}_{42}$, $\iota_{\tilde{\epsilon}_04}$, $\iota_{\tilde{\epsilon}_05}$, $\iota_{\tilde{\epsilon}_06}$ are nonnegative. So  $\beta^{(1)}_{41}\iota_{44}$, $\beta^{(1)}_{42}\iota_{54}$, $\beta^{(1)}_{41}\iota_{45}$, $\beta^{(1)}_{42}\iota_{55}$ are nonnegative.
 There are $\beta^{(1)}_{42}\iota_{54}$, $\beta^{(1)}_{42}\iota_{55}$ in $\beta^{(1)}_{41}\iota_{44}$, $\beta^{(1)}_{42}\iota_{54}$, $\beta^{(1)}_{41}\iota_{45}$, $\beta^{(1)}_{42}\iota_{55}$, and $\beta^{(1)}_{42}\iota_{54}$, $\beta^{(1)}_{42}\iota_{55}$  are positive.  Since $\beta^{(1)}_{41}\iota_{44}$, $\beta^{(1)}_{42}\iota_{54}$, $\beta^{(1)}_{41}\iota_{45}$, $\beta^{(1)}_{42}\iota_{55}$ are nonnegative, then $\iota_{77}$ and $\iota_{78}$ are positive. Therefore, the $7$-th row of matrix $Q_{et}$ is a nonzero row.

\item $\beta^{(1)}_{31},\beta^{(1)}_{42}$ are positive, and $\beta^{(1)}_{54}$ is positive, then $\beta^{(1)}_{54}\beta^{(1)}_{\epsilon_11}$ are positive.
 By calculating $L^2_{22}$,  $\beta_{52}^{(2)}=\beta_{54}^{(1)}\beta_{42}^{(1)}$. Similar to Step 2, we see that  $\beta_{52}^{(2)}$ is nonnegative. Then 
  $\beta_{52}^{(2)}\iota_{\tilde{\epsilon}_04}$, $\beta_{52}^{(2)}\iota_{\tilde{\epsilon}_05}$ are positive. By calculating $L^2_{22}L_{21}$,  $\iota_{8(10)}=\beta_{52}^{(2)}\iota_{54}$, $\iota_{8(11)}=\beta_{52}^{(2)}\iota_{55}$. So $\iota_{8(10)}$, $\iota_{8(11)}$ are positive. Similar to Step 2, it can be seen that the  $8$-th row of the matrix $Q_{et}$ is a nonzero row.
\end{enumerate}

In summary, the rows $1$ to $8$ of the matrix $Q_{et}$ are all nonzero rows.  Therefore, the rows  $1$ to $8$ of the matrix $Q$ are all nonzero rows. If the target node set $\mathcal{V}_T\subseteq\{1,2,3,4,5,6,7,8\}$, then target nodes are all reachable from the leader set. Obviously, the rows of the matrix $HQ_{et}$ are all nonzero rows, and it can be obtained that the rows of $HQ$ are also nonzero rows. Namely, there are no zero rows in the target controllability matrix $W$.
 If the unreachable node 9 is the target node, then there is a zero row in the target controllability matrix $W$, which is consistent with the conclusion of Proposition \ref{pro1}.
\end{example}

\begin{example}
In Figure \ref{fig10}, $\mathcal{V}_T=\{3,5\}$ and $\delta-$reachable nodes belong to the same cell in $\pi_{0{EP}}$. $\tilde{\mathcal{C}}_1=\{3,4\}$, $\tilde{\mathcal{C}}_2=\{5,6\}$.
There is only one target node $3$ in $\tilde{\mathcal{C}}_1$, and there is only one target node $5$ in $\tilde{\mathcal{C}}_2$.
There are no unreachable target nodes. By Theorem \ref{theo1} and Remark \ref{rem3}, system \eqref{2} is target controllable.
\end{example}

\begin{example}
From Figure \ref{fig2}, there is
\[\mathcal{P}_1= \left[ {\begin{array}{*{20}{c}}
 1&0&0\\
 0&1&-3\\
 0&1&-3\\
 0&0&0\\
 0&0&0\\
 0&1&-4\\
 0&1&-4\\
 \end{array}} \right].\]
The maximal linear independent group of rows in $\mathcal{P}_1$ can be the $1$-th, $2$-th, $6$-th  rows or the $1$-th, $2$-th, $7$-th  rows or the $1$-th, $3$-th, $6$-th  rows or the $1$-th, $3$-th, $7$-th  rows. Therefore, system \eqref{2} is target controllable if and only if target nodes are chosen from  any group of nodes $1$, $2$, $6$ or $1$, $2$, $7$ or $1$, $3$, $6$ or $1$, $3$, $7$. 
\end{example}

\begin{example}
Consider a second-order multiagent system with the same topology as Figure \ref{fig2}, we see that
\[\mathcal{A}=\left[ {\begin{array}{*{20}{c}}
0_7&I_7\\
-L&0_7 \end{array}} \right],
\mathcal{B}=\left[ {\begin{array}{*{20}{c}}
0_{7\times 1}\\ B\end{array}} \right],
\mathcal{H}=\left[ {\begin{array}{*{20}{c}}
H &0_{2\times 7} \end{array}} \right].\]
$$\begin{aligned}
rankW&=rank[\mathcal{{H}}\mathcal{{B}},\mathcal{{H}}\mathcal{{A}}\mathcal{{B}},\mathcal{{H}}\mathcal{{A}}^2\mathcal{{B}},\mathcal{{H}}\mathcal{{A}}^3\mathcal{{B}},\ldots,\mathcal{{H}}\mathcal{{A}}^{13} \mathcal{{B}}]\\
&=rank\left[0_{2\times 1},HB,0_{2\times 1},H(-L)B,\cdots,H(-L)^6B\right]\\
&=rank\left[HB,H(-L)B,\cdots,H(-L)^6B\right]\\
&=rank\left[ {\begin{array}{*{20}{c}}
0&1&-3&9&-27&81&-243\\
0&1&-4&16&-64&256&-1024
\end{array}} \right] \\
&=2.
\end{aligned}$$
Therefore,  the system \eqref{5} is target controllable in this example, which is consistent to Theorem \ref{theo3}. The simulation of the second-order multiagent system in Figure \ref{fig2} is shown in Figure \ref{fig8}.
\begin{figure}[t]
\centerline{\includegraphics[width=150pt,height=120pt]{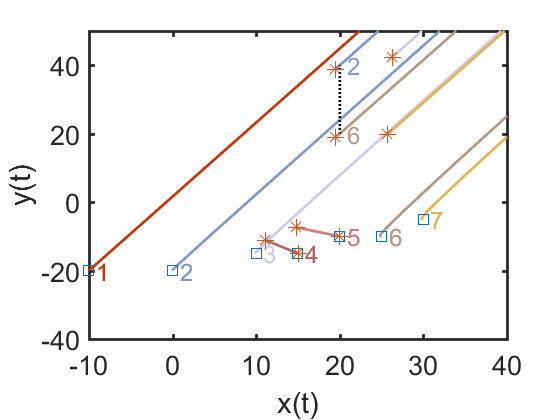}}
\caption{A  target controllable graph the second-order system.\label{fig8}}
\end{figure}
\end{example}

\section{Conclusion}\label{sec8}

For the target controllability of multiagent systems under directed weighted topology, the necessary and/ or sufficient conditions have been derived for  first-order and high-order multiagent systems. In the first-order case, the necessary and sufficient conditions for target controllability have been proposed, as well as a target node selection method to ensure the target controllability. In the high-order case, once the topology, leaders, and target nodes are fixed, the target controllability of the high-order multiagent system is shown to be the same to the first-order one. Then, the target controllability of the general linear system has been analyzed, and the graphical condition has been proposed by taking the advantage of independent strongly connected components. Future research will consider extending the work of the target controllability of general linear systems.

\newpage

\vfill

\end{document}